\documentclass[sigplan,noacm,screen]{acmart}
\AtBeginDocument{%
  }

\setcopyright{none}
\copyrightyear{2018}
\acmYear{2018}
\acmDOI{XXXXXXX.XXXXXXX}
\acmConference[Conference acronym 'XX]{Make sure to enter the correct
  conference title from your rights confirmation email}{June 03--05,
  2018}{Woodstock, NY}
\acmISBN{978-1-4503-XXXX-X/2018/06}

\renewcommand\footnotetextcopyrightpermission[1]{} 
\newcommand{\marklessfootnote}[1]{
  \begingroup
  \renewcommand\thefootnote{}
  \footnote{#1}
  \addtocounter{footnote}{-1}
  \endgroup
}

\usepackage{xcolor}

\usepackage{booktabs,multirow,tabularx,makecell,array}
\usepackage{xspace}
\usepackage{ragged2e}

\usepackage{graphicx}   

\usepackage{amsmath,amssymb}

\usepackage{algorithm}
\usepackage{algpseudocode}
\algtext*{EndFor}
\algtext*{EndWhile}
\algtext*{EndIf}
\algtext*{EndProcedure}
\algtext*{EndFunction}

\usepackage{mathtools}

\newcolumntype{Y}{>{\RaggedRight\arraybackslash}X}

\usepackage{siunitx}
\begin{document}

\title{\textsc{TridentServe}: A Stage-level Serving System for Diffusion Pipelines}

\author{
  \fontsize{11.5}{12.5}\selectfont{Yifei Xia, Fangcheng Fu, Hao Yuan, Hanke Zhang, Xupeng Miao, Yijun Liu, Suhan Ling, Jie Jiang, Bin Cui}\\
  {\fontsize{9.5}{12.5}\selectfont\emph{The Hetu Team @ Peking University}}
}

\renewcommand{\shortauthors}{Trovato et al.}


\begin{abstract}
Diffusion pipelines, renowned for their powerful visual generation capabilities, have seen widespread adoption in generative vision tasks (e.g., text-to-image/video). 
These pipelines typically follow an \emph{encode}--\emph{diffuse}--\emph{decode} three-stage architecture. 
Current serving systems deploy diffusion pipelines within a \emph{static, manual, and pipeline-level paradigm}, allocating the same resources to every request and stage.
However, through an in-depth analysis, we find that such a paradigm is inefficient due to the discrepancy in resource needs across the three stages of each request, as well as across different requests. 
Following the analysis, we propose the \emph{dynamic stage-level serving paradigm} and develop \textsc{TridentServe}, a brand new diffusion serving system.
\textsc{TridentServe} automatically, dynamically derives the placement plan (i.e., how each stage resides) for pipeline deployment and the dispatch plan (i.e., how the requests are routed) for request processing, co-optimizing the resource allocation for both model and requests.
Extensive experiments show that \textsc{TridentServe} consistently improves SLO attainment and reduces average/P95 latencies by up to 2.5$\times$ and 3.6$\times$/4.1$\times$ over existing works across a variety of workloads. 
\end{abstract}

\maketitle
\pagestyle{plain}

\section{Introduction}
\label{sec:introduction}
\textbf{\underline{Background.}} In recent years, \emph{Generative Vision Tasks}~\cite{vae,ddim,ddpm,scorebase,gan} (GVTs), such as text-to-image~\cite{sdxl,flux,fluxmodel,stablediffusion3} and text-to-video~\cite{hunyuanvideo,wan,cogvideox}, have surged in popularity, prompting a proliferation of models~\cite{attngan,stackgan,emu3,magi,sdxl,vdm} designed to address them. Among these, \emph{Diffusion Pipelines}~\cite{ddim,ddpm} have come to dominate the field for their powerful generative capabilities, giving rise to widely known pipelines such as Stable Diffusion~\cite{stablediffusion,stablediffusion3} and Sora~\cite{sora}.

\marklessfootnote{Contact: Yifei Xia (yifeixia@stu.pku.edu.cn), Fangcheng Fu (ccchengff@sjtu.edu.cn) and Bin Cui
(bin.cui@pku.edu.cn)} 

As illustrated in Figure~\ref{fig:diffusion pipeline}, \emph{Diffusion Pipelines} comprise a multi-stage inference pipeline that can be abstracted into three consecutive stages: \textbf{\emph{Encode}}, \textbf{\emph{Diffuse}}, and \textbf{\emph{Decode}}. The \emph{Encode} stage embeds the user guidance (e.g., the text prompt) into a \emph{condition} $c$ and passes it to the \emph{Diffuse} stage~\cite{t5xxl,llama3}. Then, the \emph{Diffuse} stage samples Gaussian noise~\cite{ddpm} corresponding to the targeted output resolution (and duration for video) in the low-dimensional latent space~\cite{stablediffusion}, concatenates it with $c$, and then performs multi-step forward propagation for denoising through the \emph{Diffusion model} to obtain a latent result in the low-dimensional space~\cite{stablediffusion3,sdxl}. Finally, the latent result is fed into the \emph{Decode} stage to obtain the visual result in pixel space~\cite{vae}.
With the increasing widespread adoption of \emph{Diffusion Pipeline}, how to serve such pipelines efficiently has become a timely and important topic~\cite{xdit,videosys}.
\begin{figure}[!t]
    \centering
    \includegraphics[width=0.95\linewidth]{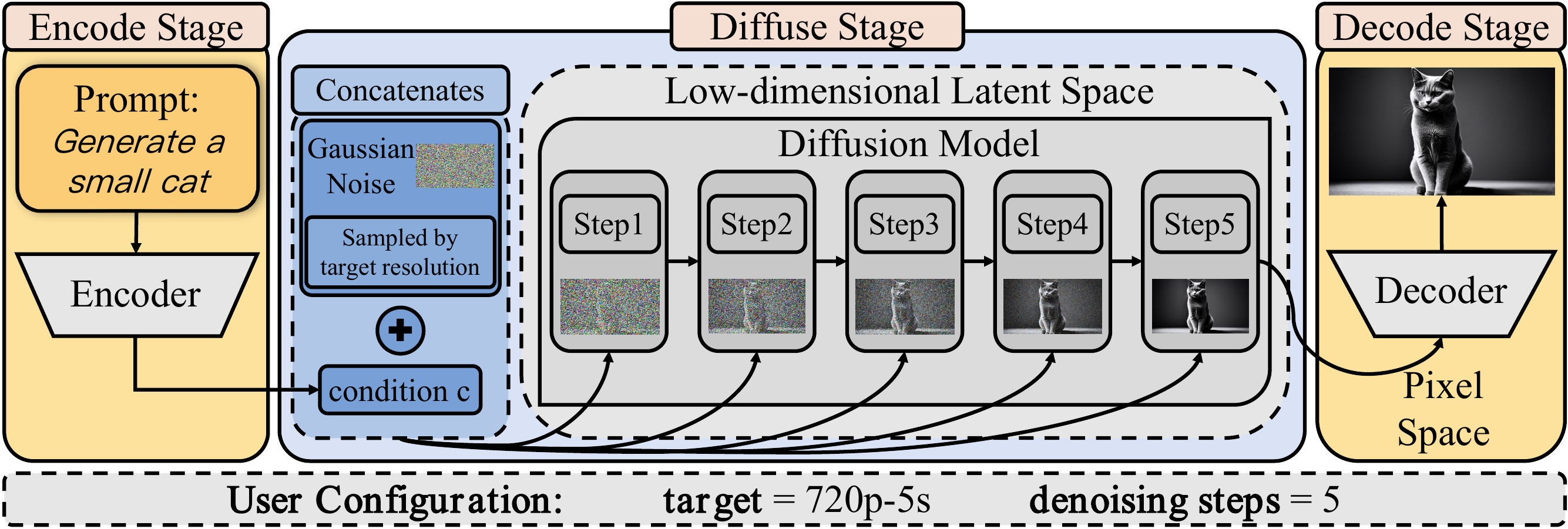}
    \caption{\small{The typical inference process of a \emph{Diffusion Pipeline}.
    }}
    \label{fig:diffusion pipeline}
\end{figure}

\textbf{\underline{The Current Landscape.}} Currently, predominant practices~\cite{xdit,videosys} of \emph{Diffusion Pipeline} serving follow two designs:
\begin{itemize}
    \item \textit{\underline{Static, manual resource allocation for stage models.}} Current systems manually allocate resources for the stage models under a static paradigm: 
    they either co-locate all stages on each GPU (for small pipelines) or statically deploy the pipeline in a disaggregated fashion and assign a fixed number of replicas to different stages (for large pipelines).
    
    \item \textit{\underline{Static, pipeline-level resource allocation for requests.}} When handling different requests, current systems assign a fixed number of GPUs to every request (a.k.a., use the same parallelism~\cite{ulysses,ring,tp,pp} strategy). 
    Meanwhile, for any single request, the different stages also share the same amount of allocated GPU resources at the pipeline-level.
\end{itemize}

However, after our detailed analysis of the architecture and workloads of \emph{Diffusion Pipelines} (\S\ref{sec:analysis}), we find that such a na\"ive, static design is inefficient for the following reasons:

First, \emph{Diffusion Pipelines} are both \emph{stage-heterogeneous} and \emph{request-heterogeneous}, rendering this naive resource allocation inefficient and wasteful. 
On the one hand, GVTs exhibit substantial workload variation across requests due to differences in target resolutions or video durations, and different requests manifest distinct scalability characteristics (\S\ref{sec:analysis}). Hence, statically allocating identical resources to all requests is inefficient, as illustrated in Figure~\ref{fig:method_comparsion}(a). On the other hand, stages within the same \emph{Diffusion Pipeline} have distinct architectures and behaviors, and therefore exhibit different scalability (\S\ref{sec:analysis}). Thus, when processing one request, allocating the same resources to all three stages of that request at the pipeline-level is likewise inefficient, as shown in Figure~\ref{fig:method_comparsion}(b). Finally, this heterogeneity varies across different \emph{Diffusion Pipelines}, making manual deployment extremely difficult and brittle.

Second, production GVT workloads exhibit complex and highly dynamic arrival patterns~\cite{shahrad2020serverless,diffserve}. 
Particularly, in \S\ref{sec:analysis}, we analyze that as workload arrival patterns shift, the resources required by different stages also change, thereby rendering static deployments maladaptive. Besides, the change in workload arrival patterns also necessitates adjusting the parallelism choices of different requests. Unfortunately, current approaches fail to handle such dynamicity.

\begin{figure}[!t]
    \centering
    \includegraphics[width=0.9\linewidth]{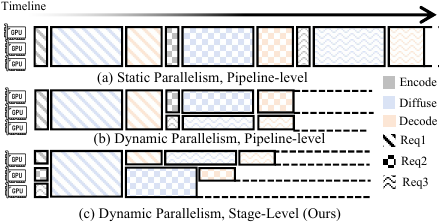}
    \caption{\small{Example of different serving methods. Assuming the GPU demands for the \emph{Diffuse} stage of requests \#1/2/3 are 3,2,1 respectively, and the demand for all \emph{Encode}, \emph{Decode} stages is 1.
    }}
    \label{fig:method_comparsion}
\end{figure}

Therefore, to serve diffusion pipelines efficiently, we expect a system capable of fine-grained, dynamic resource allocation for both request processing and model deployment. However, realizing such a system presents two challenges:
\begin{itemize}
    \item \textit{How to efficiently support such fine-grained, dynamic processing and deployment?} All existing diffusion serving systems, such as xDiT~\cite{xdit} and VideoSys~\cite{videosys}, are coarse-grained and pipeline-level.
    Once varying resources are allocated to different requests and stage models, how to carry out the serving process efficiently is an unexplored question.
    \item \textit{How to accurately solve dynamic resource allocation?} For resource allocation and processing of each request, the inter-stage topology and communication, tight GPU-memory budgets~\cite{rtx4090,a100,l40}, and divergent resource demands 
    together vastly expand the decision space. 
    Meanwhile, for resource allocation of stage models, differing requirements across pipelines and dynamically shifting workload arrival patterns~\cite{shahrad2020serverless,diffserve} also exacerbate the complexity of solving the problem.
\end{itemize}

\textbf{\underline{Our Solution.}} To address this, we propose \textsc{TridentServe}, the first system that supports dynamic, stage-level resource allocation for both request processing and model deployment.
To facilitate stage-level processing and deployment, \textsc{TridentServe} introduces two abstractions:
it casts model-side allocation as \emph{placement plans}, enabling automatic, dynamic deployment; and casts request-side resource allocation as \emph{dispatch plans}, enabling dynamic, stage-level resource allocation (Figure~\ref{fig:method_comparsion} (c)) for request processing. 
The \emph{Runtime Engine} of \textsc{TridentServe} is in charge of the efficient execution of both plans.
For the \emph{placement plan}, it automatically deploys the pipeline at initialization and performs an \emph{Adjust-on-Dispatch} mechanism that integrates stage-model resource transitions into request processing, allowing the system to flexibly adapt to diverse workload patterns while achieving seamless, no-downtime deployment switch.
For the \emph{dispatch plan}, we define an atomic three-step procedure to execute requests at the stage-level efficiently.

As for plan generation, we 1) employ a \emph{Dynamic Orchestrator} for \emph{placement plan} generation, which automatically optimizes model deployment given the model and workload statistics without manual tuning;
and 2) propose a two-step \emph{Resource-Aware Dispatcher} to generate \emph{dispatch plans}, which achieves efficient resource utilization while greatly improving overall system efficiency.

We summarize our main contributions as follows:
\begin{itemize}
  \item We present the first systematic, stage-aware analysis of diffusion pipelines, revealing the asymmetry in resource demand across stages and requests, which motivates a \emph{stage-level} serving paradigm.
  \item We develop \textsc{TridentServe}, the first \emph{Diffusion Pipeline} serving system with \emph{dynamic, stage-level} resource allocation for both models and requests. It integrates a \emph{Dynamic Orchestrator} and a \emph{Resource-Aware Dispatcher} that co-optimize the model placement and request dispatching, markedly improving serving efficiency and robustness under complex workload arrival patterns.
  \item \textsc{TridentServe} auto-deploys pipelines to eliminate laborious manual tuning. Extensive experiments show that it consistently improves SLO attainment, average/P95 latencies by up to 2.5$\times$, and 3.6$\times$/4.1$\times$ against the strongest baseline under a variety of workloads.
\end{itemize}

\section{Preliminary}
\label{sec: preliminary}
This section introduces the preliminary literature of our work. Frequently used notations are listed in Table~\ref{tab:symbols}.

\subsection{Diffusion Pipelines}

As sketched in \S\ref{sec:introduction} and Figure~\ref{fig:diffusion pipeline}, \emph{Diffusion Pipelines} comprise three stages—\emph{Encode} ($E$), \emph{Diffuse} ($D$), and \emph{Decode} ($C$). Below, we summarize each stage’s architecture and workload.

\textbf{Architecture and Workload by Stage.} The stages differ markedly in model arch and processing length~\cite{hunyuanvideomodel,cogvideoxmodel,fluxmodel,sd3model}.

\emph{Encode} uses a Transformer-based~\cite{transformer} \emph{Encoder}~\cite{t5xxl} to embed lightweight guidance, with processing length $l^E_{proc}\!\le\!500$.

\emph{Diffuse} is the core of \emph{Diffusion Pipeline}, which iteratively denoises latent Gaussian noise $\mathbf{x}_T\!\sim\!\mathcal{N}(\mathbf{0},\mathbf{I})$ for $T$ steps under the embedded condition $c$~\cite{ddpm,ddim}. At step $t$, it updates
\begin{equation*}
  \epsilon_t = \epsilon_\theta(\mathbf{x}_t,\,t,\,c), \qquad
  \mathbf{x}_{t-1} = \Phi(\mathbf{x}_t,\,t,\,\epsilon_t)\,,
\end{equation*}
where $\epsilon_\theta$ is the \emph{Diffusion model}, $\epsilon_t$ the predicted noise, and $\Phi$ the denoising function. Its processing length $l^D_{proc}$ spans from hundreds to $10^5$~\cite{stablediffusion3,fluxmodel,hunyuanvideomodel}; it is compute-bound and typically consumes $>\!70\%$ of end-to-end time~\cite{sanz_howsdxl_2024} (\S\ref{subsec: dispatch plan}).

\emph{Decode} uses a memory-bound \emph{Decoder}~\cite{vae} to map the latent output to pixel output. It usually accounts for 15\%–30\% of execution time and can exhibit large activation-memory~\cite{hf_autoencoderkl_docs}.

Table~\ref{tab:diffusion pipelines} summarizes representative pipelines: Stable Diffusion 3~\cite{sd3model} and Flux.1~\cite{fluxmodel} (image), CogVideoX1.5-5B~\cite{cogvideoxmodel} and HunyuanVideo~\cite{hunyuanvideomodel} (video), along with their stage-level traits. These architectural and workload distinctions directly induce heterogeneous resource demands, as detailed in \S\ref{sec:analysis}.

\begin{table}[!t]
  \footnotesize
  \setlength{\tabcolsep}{4pt}
  \caption{\small{Frequently used notations.}}
  \begin{tabularx}{\columnwidth}{@{}lX@{}}
    \toprule
    \textbf{Notation} & \textbf{Meaning} \\ \midrule
    $E,\,D,\,C$                  & abbreviation for \textit{Encode}, \textit{Diffuse} and \textit{Decode} stage \\
    $\mathcal{G},\,G,\,g$        & set of all GPUs, number of all GPUs, GPU index \\
    $\mathcal{R},\,r$        & set of pending requests, request index \\
    $\mathcal{S},\,s$            & $\mathcal{S} = [E, D, C]$ is stages list, $s \in \mathcal{S}$ is one of the stage \\
    $k$                          & parallel degree (specifically refers to SP degree), $k \in \{1,2,4,8\}$ \\
    $l^s_{proc}$              & the processing sequence length for stage $s$. \\
    \bottomrule
  \end{tabularx}
  \label{tab:symbols}
\end{table}

\begin{table}[!t]
    \centering
    \caption{\small{Typical Diffusion Pipelines. 
    We abbreviate Stable-Diffusion-3-medium as Sd3, CogVideoX1.5-5B as Cog, HunyuanVideo as HYV; 
     AE-KL-Cog/HYV denote model-specific variants of AutoEncoderKL. Model sizes in billions (B).}}
    \label{tab:diffusion pipelines}
    \begingroup
    \setlength{\tabcolsep}{4pt}
    \renewcommand{\arraystretch}{1.1}
    \resizebox{\linewidth}{!}{%
    \begin{tabular}{@{}l l c c l c c l c c@{}}
        \toprule
        \multirow{2}{*}{Pipeline} &
            \multicolumn{3}{c}{Enc.} &
            \multicolumn{3}{c}{Diff.} &
            \multicolumn{3}{c}{Dec.} \\
        \cmidrule(lr){2-4}\cmidrule(lr){5-7}\cmidrule(lr){8-10}
        & Name & B & $l_{\text{proc}}$ & Name & B & $l_{\text{proc}}$ & Name & B & $l_{\text{proc}}$ \\
        \midrule
        Sd3 &
            T5-XXL & 4.8 & 30--500 &
            Sd3-DiT & 2 & 100--60k &
            AE-KL & 0.1 & 100--60k \\
        Flux &
            T5-XXL & 4.8 & 30--500 &
            Flux-DiT & 12 & 100--60k &
            AE-KL & 0.1 & 100--60k \\
        Cog &
            T5-XXL & 0.35 & 30--500 &
            Cog-DiT & 4.2 & 1k--120k &
            AE-KL-Cog & 0.45 & 1k--120k \\
        HYV &
            Llama3-8B & 8 & 30--500 &
            HYV-DiT & 13 & 1k--120k &
            AE-KL-HYV & 0.5 & 1k--120k \\
        \bottomrule
    \end{tabular}}
    \endgroup
\end{table}

\subsection{Resource Allocation in \emph{Diffusion Pipeline}}
In \emph{Diffusion Pipeline} serving, resource allocation has two facets: the \emph{model} side and the \emph{request} side.

\textbf{Resource Allocation for Models.}
This is primarily achieved via \emph{model replication} and \emph{model parallelism}~\cite{tp,pp} (MP), which jointly determine how models are placed and how many replicas are provisioned. 
For \emph{model replication}, the goal, especially under \emph{disaggregated} deployments~\cite{nvidia-triton-sd-pipeline}, is to choose per-stage replica counts and placements to balance stage throughput while minimizing inter-stage communication.
\emph{Model parallelism} partitions parameters (layer/tensor) across devices to relieve computational and memory pressure. However, as will be shown in \S\ref{sec:analysis}, it is typically less efficient for \emph{Diffusion Pipelines}~\cite{pipefusion,distrifusion}. 
Hence, a de facto approach is to configure the MP degree to be the smallest number of GPUs that fits the model.\footnote{For clarity of discussion, we mainly elaborate on the case when MP is not employed throughout this work, 
and introduce how to integrate our approach with MP in Appendix~\ref{subsec: extend model parallelism}.
}

\textbf{Resource Allocation for Requests.}
For each request, it's possible to reduce inference latency by allocating more computing resources and splitting the sequence across devices.
However, as we will show in \S\ref{sec:analysis}, determining the resource allocation for each request (i.e., sequence parallelism~\cite{ulysses,ring} (SP) degree) also requires considering the stage heterogeneity.

\section{Analysis and Motivation}
\label{sec:analysis}
We analyze how parallelism strategies and workload arrival patterns affect each stage, and introduce the key motivations for our system design.

\textbf{Impact of Parallelism Strategies.}
The \emph{Encode} stage has $l^E_{proc}\!\le\!500$ and rarely benefits from parallelism; instead, it primarily gains from batching (Appendix~\ref{subsec: extend batch size}), so we focus on \emph{Diffuse} and \emph{Decode}. Using Flux.1~\cite{flux} as an example (more models in Appendix~\ref{sec: other impact effects}), we observe markedly different scaling behavior for the \emph{Diffusion model} and the \emph{Decoder}. Figure~\ref{fig:parallel effect} reports speedups as the degree of parallelism varies: at higher target resolutions, larger degrees are advantageous, whereas at lower resolutions, smaller degrees suffice. At a fixed resolution, the two stages also scale differently. \emph{Diffuse} scales better than \emph{Decode} because decoding is memory\mbox{-}bound~\cite{vae,hf_sd3_memopt_2024}, so increasing the degree yields limited gains.
Meanwhile, MP exhibits worse scalability than SP consistently, suggesting that it is preferable to use a small MP degree and scale by SP.

\textit{\underline{Insight 1}}: Resource needs (parallel degree) vary across requests as well as across stages within one request.

\begin{figure}[!t]
    \centering
    \includegraphics[width=\linewidth]{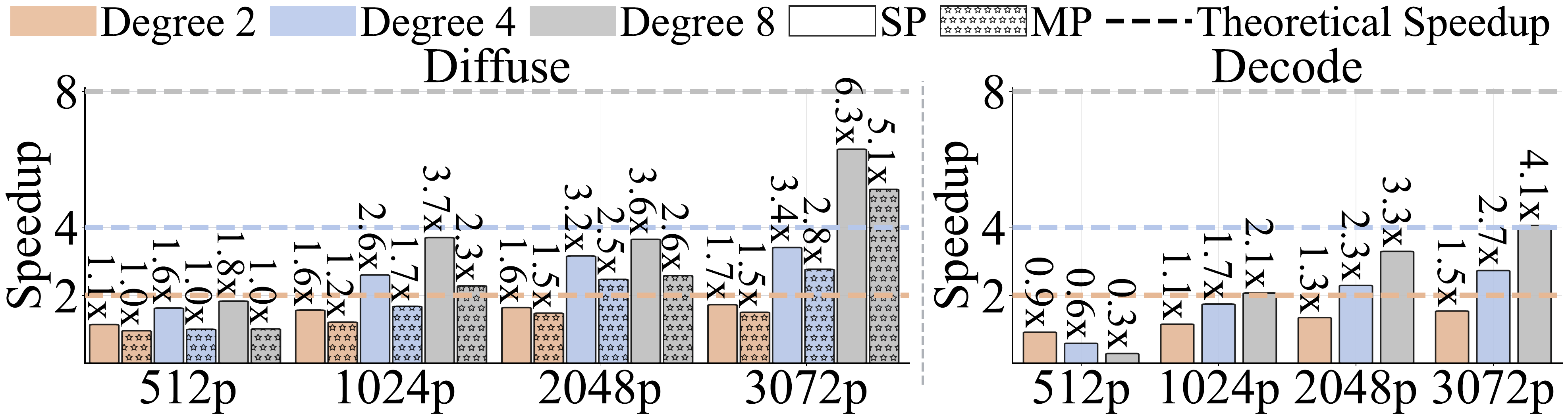}
    \caption{\small{Parallelism effects on \emph{Diffuse} and \emph{Decode} stages of Flux.1, tested on NVIDIA L20.}}
    \label{fig:parallel effect}
\end{figure}

\textbf{Impact of Workload Arrival Pattern.}
Due to stages having different processing speeds, we must provision replicas asymmetrically to equalize throughput and prevent stage congestion, especially when using disaggregated deployment. 
However, as workload arrival patterns vary, the replica counts that keep stage speeds balanced must also adapt.
Using Flux.1 (Figure~\ref{fig:workload pattern}), we profile per-stage demand across patterns and arrival rates and report the replica proportions required for throughput balance. These proportions shift continuously because stages differ in load elasticity.

\textit{\underline{Insight 2}}: Resource demands (stage model replicas) vary as the workload arrival pattern varies.

\begin{figure}[!t]
    \centering
    \includegraphics[width=\linewidth]{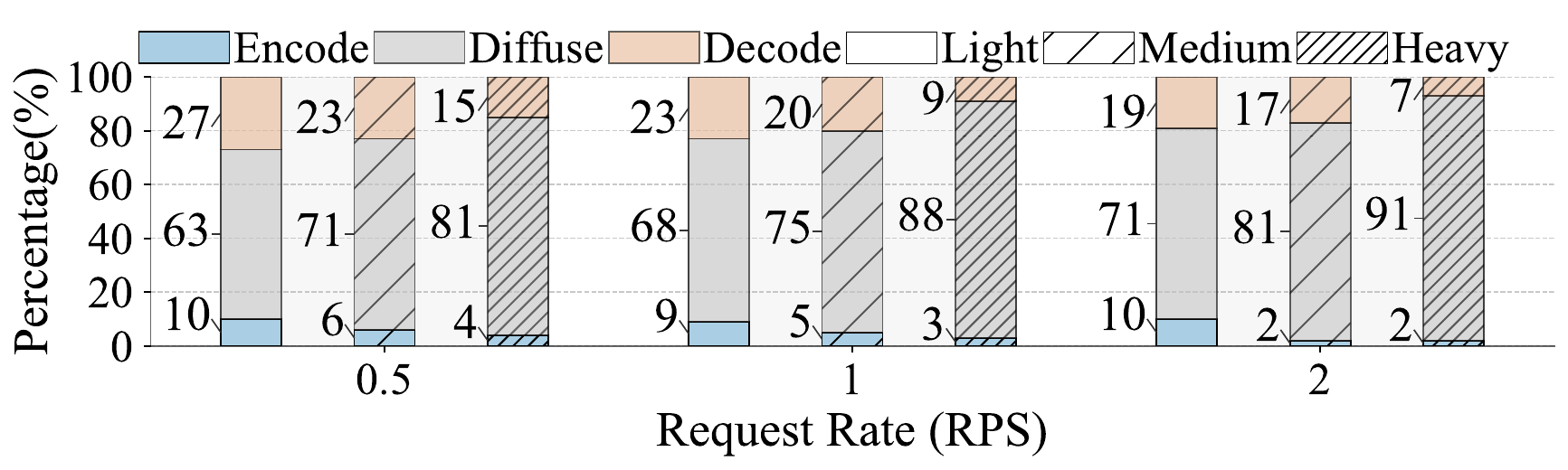}
    \caption{\small{Model replica demands to achieve a balanced processing speed, tested on NVIDIA L20. Light/Medium/Heavy denote workloads defined in \S\ref{subsec: experiment setup}
    }}
    \label{fig:workload pattern effect}
\end{figure}

Based on the above analysis, we distill two design principles for efficient \emph{Diffusion Pipeline} serving.

\textbf{Principle 1: Stage-level Dynamic Resource Allocation for Requests.} 
Because scalability varies both across requests and across the stages within a single request, resource needs differ at the stage-level. Accordingly, allocation should be dynamic and \emph{stage-level} for each request, rather than static and \emph{pipeline-level}.

\textbf{Principle 2: Dynamic Resource Allocation for Stage Models.} 
In real-world scenarios, workload arrival patterns fluctuate over time, and the resource demands of individual stages shift accordingly. Moreover, different pipelines have distinct architectures (Table~\ref{tab:diffusion pipelines}), leading to heterogeneous deployment needs. Therefore, the system should automatically deploy and dynamically allocate resources to stage models, adapting to pipeline specifics and workload dynamics rather than relying on brittle, static manual configurations.

\textbf{Our Design.}
Motivated by these, we design a unified serving system, \textsc{TridentServe}, that supports \emph{Stage-level, dynamic resource allocation}  for both requests (\S\ref{subsec: dispatch plan}) and stage models (\S\ref{subsec: placement plan}),  which is all backed by a \emph{Runtime Engine} (\S\ref{sec:engine}).

\section{\textsc{TridentServe} Overview}
\label{ch:overview}
\begin{figure}[!t]
    \centering
    \includegraphics[width=\linewidth]{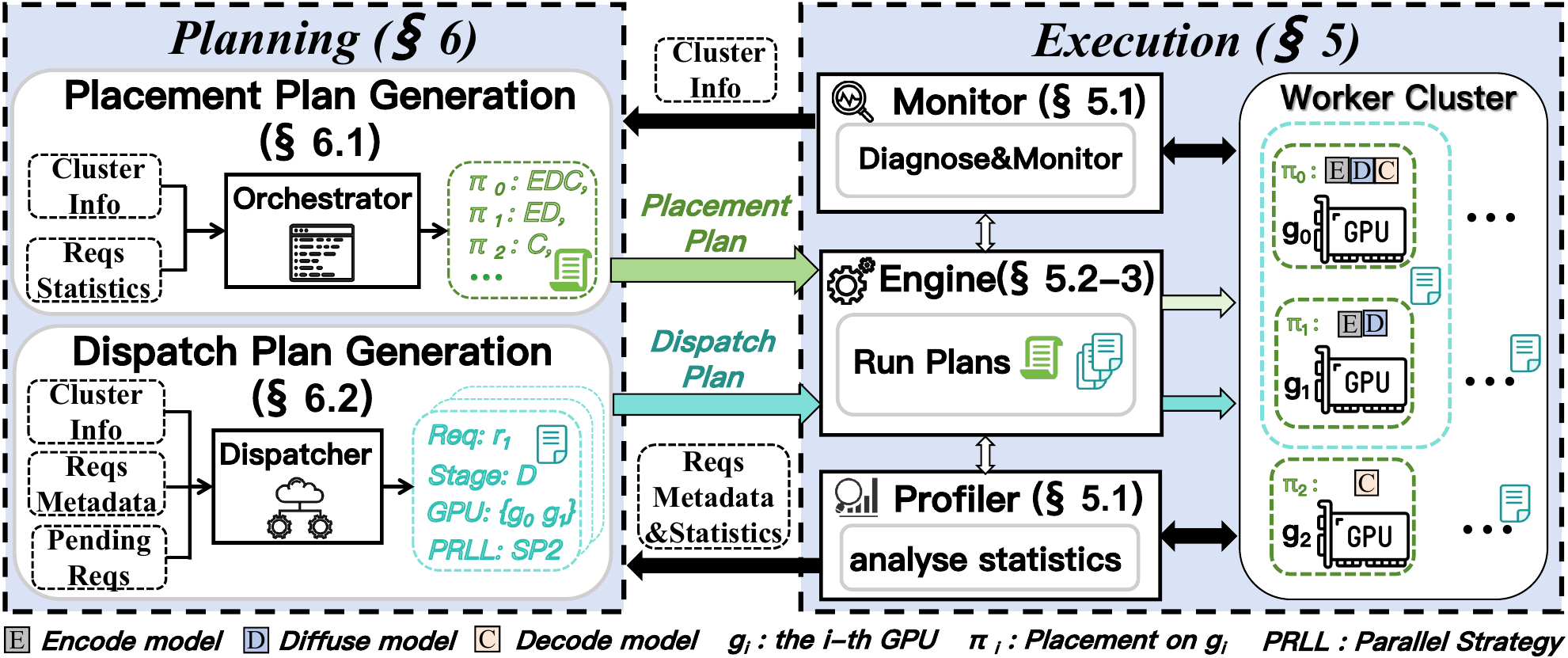}
    \caption{\small{Overview of \textsc{TridentServe}. }
    }
    \label{fig:overview}
\end{figure}

The overview of \textsc{TridentServe} is shown in Figure~\ref{fig:overview}. The system is structured around two tightly coordinated phases, \emph{planning} and \emph{execution}:

\emph{Planning}: As shown in Figure~\ref{fig:overview}, taking cluster information and request statistics as input, a dynamic \textbf{Orchestrator} generates a \emph{placement plan} to guide the resource allocation for stage models, where $\mathcal{P}=\{\small \pi_{g} \mid  g \in [0,G-1]\}$ is the placement plan for the worker cluster and $\pi_g = \{\langle s_1,\dots,s_k\rangle \mid 1\le k\le 3,\; s_i\in\mathcal{S}\,\}$ is the placement of g-th GPU. For example, $\pi_2=\langle E\;D \rangle$ states that GPU~\#2 hosts the replica of stages \emph{Encode} and \emph{Diffuse}.
A resource-aware \textbf{Dispatcher} takes the cluster information of current placement and the request metadata as input, producing \emph{dispatch plans} $\{\Gamma^s_r\}$ that, for each dispatched request $r$'s stage $s$, specify the GPU set and parallelism strategy used to execute it, thereby enabling stage-level resource allocation for requests.

\emph{Execution}: \textsc{TridentServe} features a high-performance \textbf{Runtime Engine} that executes the corresponding plans efficiently.
It is assisted by a comprehensive \textbf{Profiler} and a lightweight \textbf{Monitor}, which continuously collect run-time statistics (e.g., cluster information, request statistics, and metadata) and inform the generation of $\Gamma$ and $\mathcal{P}$.

\subsection{Overall Routine}
\label{subsec:routine}
Algorithm~\ref{alg:workflow} shows the high-level workflow of \textsc{TridentServe}:

\begin{itemize}
  \item[(1)] \textbf{Offline Profiling.}  
    The \emph{Profiler} analyzes each candidate resolution (and duration) to collect latency–memory statistics, etc.\ (Line~1).

  \item[(2)] \textbf{Bootstrap Placement.}
    When the system starts, the \emph{Orchestrator} takes the profiling results as input to generate an initial placement plan $\mathcal{P}_{init}$ (\S\ref{subsec: placement plan}), which the \emph{Runtime Engine} then materialises across the GPU cluster (Lines~2–3).

  \item[(3)] \textbf{Online Serving.}  
    During serving, under the current $\mathcal{P}$, the \emph{Dispatcher} combines request metadata with cluster information reported by the \emph{Monitor} and \emph{Profiler} to produce dispatch plans set $\{\Gamma^s_i\}$ (\S\ref{subsec: dispatch plan}) and sends them to the \emph{Runtime Engine} for execution (\S\ref{sec:engine}) (Lines~9–10). 

  \item[(4)] \textbf{Adaptive Re-Placement.} 
    The \emph{Monitor} runs continuously. 
    Once it detects arrival pattern changes causing congestion (detailed in \S\ref{subsec: execution placement plan}), the \emph{Orchestrator} promptly derives a new placement plan $\mathcal{P}_{switch}$ (\S\ref{subsec: placement plan}), and the \emph{Runtime Engine} applies it via an \emph{Adjust-on-Dispatch} procedure, enabling live re-deployment without downtime (\S\ref{sec:engine}) (Lines~6–8).
\end{itemize}

This routine sketches how, through the abstractions of $\Gamma$ and $\mathcal{P}$ and an efficient \emph{Runtime Engine}, the system supports dynamic, stage-level request processing, as well as automatic dynamic allocation of stage replicas. Then \S\ref{sec:engine} and \S\ref{sec:planner} detail plan execution by the \emph{Runtime Engine} as well as the concrete methods used to generate $\mathcal{P}$ and $\Gamma$.
\begin{algorithm}[t]
\caption{High-level Workflow of \textsc{TridentServe}}
\label{alg:workflow}
\begin{algorithmic}[1]
\Require Profiling data \(\mathcal{D}\); GPU set \(\mathcal{G}\).
\State \(\mathrm{ReqsInfo}\gets \mathrm{Profiler}(\mathcal{D})\)
\State \(\mathcal{P}_{init}\gets \mathrm{Orchestrator}(\mathrm{ReqsInfo},\mathcal{G})\)
\State \(\mathrm{RuntimeEngine(\mathcal{P}_{init}, \mathcal{G})}\)
\While{running}
  \State \(\mathcal{R} \gets \mathrm{new \ requests}\)
  \If{\(\mathrm{PatternChange}(\mathrm{Monitor},\mathcal{R},\mathcal{D},\mathcal{P},\mathcal{G})\)}
    \State \(\mathcal{P}_{switch}\gets \mathrm{Orchestrator}(\mathrm{ReqsInfo}, \mathcal{R}, \mathcal{G})\)
    \State \(\mathrm{RuntimeEngine(\mathcal{P}_{switch}, \mathcal{G})}\)
  \EndIf
  \State \(\{\Gamma_{i}\} \gets \mathrm{Dispatcher(\mathcal{R}, \mathcal{P},\mathcal{G}, \mathrm{ReqsInfo})}\)
  \State \(\mathrm{RuntimeEngine}(\{\Gamma_{i}\}, \mathcal{G})\)
  
\EndWhile
\end{algorithmic}
\end{algorithm}

\begin{figure*}[t]
    \centering
    \includegraphics[width=\linewidth]{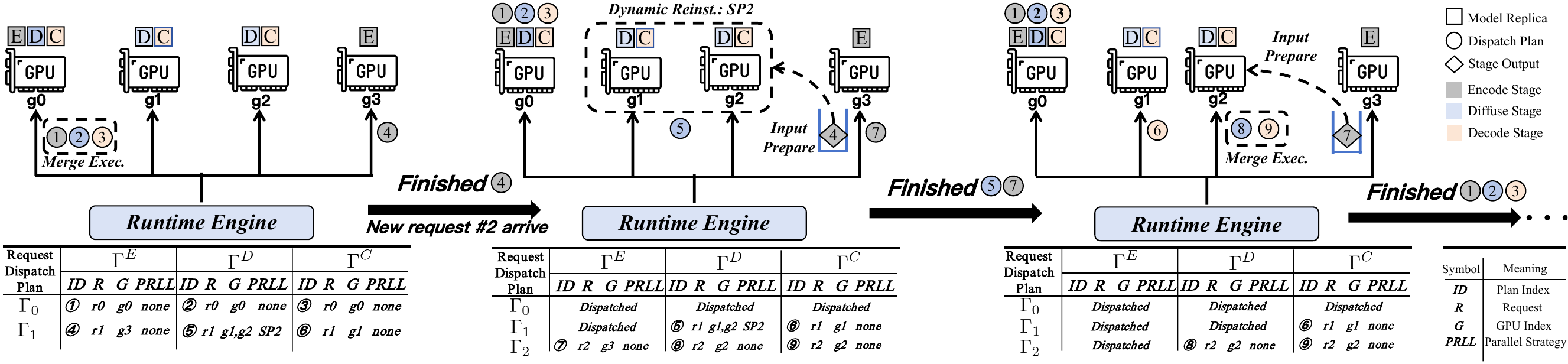}
    \caption{\small{Example of the Execution of Dispatch Plans.
    }}
    \label{fig:dispatch example}
\end{figure*}

\section{Plan Execution}
\label{sec:engine}
In this section, we introduce the efficient execution of $\mathcal{P}$ and $\Gamma$ enabled by the \emph{Runtime Engine}. We begin by describing the \emph{Monitor} and \emph{Profiler}, and then provide a detailed account of the \emph{Runtime Engine}'s execution procedure.

\subsection{Profiler and Monitor}
\label{sec:engine-profmon}

\textbf{Profiler.}
Leveraging the strong predictability~\cite{peebles2023scalable,hf_sd_speed_memory} of execution time and memory
footprint in GVT workloads, \textsc{TridentServe} provides a comprehensive offline profiler.  
On pre-supplied data, the profiler collects information at two levels of granularity:

\begin{itemize}
  \item \textbf{Request Metadata.}  
        For each stage of a range of requests,
        it measures the expected running time and peak activation memory under each parallel strategy.
  \item \textbf{Request Statistics.}  
        Statistical information includes the processing length and peak memory distribution of all stages, which will be used to guide the \emph{Orchestrator} and \emph{Dispatcher} in generating plans (\S\,~\ref{sec:planner}).
\end{itemize}

\textbf{Monitor.}
\emph{Monitor} runs in a periodic, clock-driven behavior, observes the GPU worker cluster information to furnish real-time planning feedback along two dimensions:

\begin{itemize}
  \item \textbf{GPU-worker Status.}  
        For each GPU $g$, the \emph{Monitor} reports whether it is idle, its current placement $\pi_g$, and residual memory capacity.
        If the worker is busy, it also reports the dispatch plan in execution, the start timestamp, and the estimated runtime.
  \item \textbf{Stage Throughput.}  
        The \emph{Monitor} records the processing rate for each placement type, $v_{\pi}$, enabling the \emph{Orchestrator} to rebalance placements when throughput becomes skewed.
\end{itemize}

\subsection{Execution of \emph{Dispatch Plan}}
The \emph{Resource-Aware Dispatcher} triggers in a clock-driven manner and generates the \emph{dispatch plans} to the \emph{Runtime Engine} for execution.
For each \emph{dispatch plan}, once its required GPU set is idle and the predecessor stage’s plan has completed, the \emph{Runtime Engine} reserves the resources and executes it in three steps: \emph{Dynamic Reinstance}, \emph{Stage Prepare}, and \emph{Merging Execute}. Below, we first introduce the three steps and then illustrate a concrete example in Figure~\ref{fig:dispatch example}.

\textbf{Dynamic Reinstance.}
Given a \emph{dispatch plan}, the \emph{Runtime Engine} temporarily groups the targeted GPU workers into an execution instance and activates the requisite communication group.
To avoid per-dispatch setup overhead and the excessive buffer footprint that would result from pre-initializing a communication group for every parallel strategy~\cite{pytorch_fsdp_tp_tutorial}, we prepare only a small hot set of \emph{intra-machine worker combinations} and reuse a single communication buffer within each combination.
Other infrequent combinations are lazily initialized on first use.
This hot-set and lazy-init design enables millisecond-scale reconfiguration without global pauses and keeps the memory footprint bounded, preventing communication-buffer accumulation.

\textbf{Stage Preparation.}
After forming the execution instance, the Runtime Engine prepares (i) the resident model replica for the target stage and (ii) the stage inputs.
If the \emph{placement plan} has recently changed such that the required stage replicas are \emph{temporarily} not resident on the selected workers, we invoke \emph{Adjust-on-Dispatch} at this step (see \S\ref{subsec: execution placement plan} for details).

Since dispatch is performed at the stage-level, the inputs for the current \emph{dispatch plan} may not reside on the selected GPU set. To curb additional input-transfer latency, we employ a \emph{proactive push} scheme that maximizes compute--communication overlap without inflating device memory. Concretely, each GPU worker maintains a device-resident \emph{handoff buffer} (HB) for temporarily staging inter-stage tensors. When a preceding dispatch plan completes and a successor exists, the preceding stage proactively enqueues its outputs into the successor stage's HB so the successor can read them directly at launch; if the successor stage is still computing, the transfer overlaps its computation, as shown in Figure~\ref{fig:overlap}.

To remain OOM-safe under bursts, every HB has a capacity limit \(Cap_{hb}\).
If the destination HB is full, tensors are buffered to pinned host memory, and the push completes via the host path; the successor stage then reads directly from the pinned host buffer at launch.

GPU transfers follow a two-step, locality-aware policy:
\emph{(i) inter-node}: use GPUDirect RDMA~\cite{nvidia_gpudirect_rdma_overview,nvidia_blog_gpudirect_rdma_benchmark} to send to one worker in the destination set, then intra-set broadcast via the pre-initialized communicator;
\emph{(ii) intra-node}: directly broadcast via the shared communicator, avoiding unnecessary copies.

\textbf{Merging Execute.}
Once preparation is complete, the \emph{Runtime Engine} invokes the corresponding model to execute the target stage.
To further improve efficiency, we \emph{merge} consecutive $\Gamma^s_r$ for the same request $r$ that target an identical GPU set, executing them as a single atomic run to eliminate redundant CPU-side scheduling.

\begin{figure}[!t]
    \centering
    \includegraphics[width=0.9\linewidth]{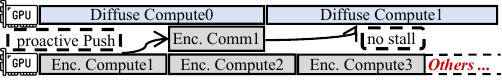}
    \caption{\small{The \emph{Compute} and \emph{Comm} Overlap of \emph{proactive push}.}}
    \label{fig:overlap}
\end{figure}

\textbf{Example.} In Figure~\ref{fig:dispatch example}, nine \emph{dispatch plans} across three requests are executed as follows. Initially, only the dispatch plans for requests \#0 and \#1 have arrived, and the Engine must complete each request’s $\Gamma_r^{E}$ before subsequent stages, so it would run \textcircled{1} and \textcircled{4}; however, the Engine detects that \textcircled{1}\textcircled{2}\textcircled{3} are on the same EDC–resident worker and therefore \emph{merging executes} them. At the second tick, \textcircled{4} finishes and a new request \#2 arrives; the predecessor of \textcircled{5} has completed and GPUs \#1/\#2 are idle, so the Engine first performs \emph{Dynamic Reinstance} to form an SP2 instance on the two DC–resident GPUs, then, during \emph{Stage Preparation}, receives from GPU \#3 the output of the prior plan \textcircled{4}, and runs \textcircled{5} on this instance; meanwhile, request \#2’s plan \textcircled{7} can also execute on GPU \#3. At the third tick, \textcircled{5} and \textcircled{7} complete, GPUs \#2/\#3/\#4 are idle, and the predecessors of \textcircled{6} and \textcircled{8} are done, so they are executed; for \textcircled{8}, it \emph{merge executes} with \textcircled{9} on the same worker, and during \emph{Stage Preparation} it receives from GPU \#3 the output of \textcircled{7}.

\subsection{Execution of \emph{Placement Plan}}
\label{subsec: execution placement plan}
\textsc{TridentServe} applies \emph{placement switches} with an \emph{Adjust-on-Dispatch} policy: the system updates placement \emph{metadata} immediately, while actual replica changes are deferred to the next dispatch that needs them. This enables seamless transitions without downtime.

\textbf{Stage Replica Management.}
Each node maintains \emph{one shared CPU replica per stage}. GPU workers host only the stages assigned by the current placement.

\textbf{Adjust-on-Dispatch.}
The \emph{Monitor} tracks per-stage throughput over a sliding window of length $T_{\text{win}}$\footnote{We set $T_{\text{win}}$ per \emph{Diffusion Pipeline} due to differing average processing speeds of different pipelines; detailed settings see Appendix~\ref{subsec: workload setting}}. If the fastest stage is at least \(1.5\times\) the slowest, the \emph{Dynamic Orchestrator} derives a new placement plan \(\mathcal{P}_{\mathrm{switch}}\). \emph{Runtime Engine} immediately updates \(\mathcal{P}_{\mathrm{switch}}\) in its placement metadata; no replicas are actually changed yet. The \emph{Dispatcher} generates new \emph{dispatch plans} according to \(\mathcal{P}_{\mathrm{switch}}\). 
Since every worker executes plans in a FIFO strategy, the in-flight/queued dispatches created under the old placement will finish as planned before the new ones, which ensures the safety of \emph{Adjust-on-Dispatch} and will not cause erroneous execution. 
When a dispatch plan based on \(\mathcal{P}_{\mathrm{switch}}\) finds that the required stage is not resident on GPU, \emph{Runtime Engine} loads the needed stage replica on \emph{Stage Preparation} of the dispatch plan through a two-step efficient transfer manner: 1) it first tries an intra-node GPUDirect P2P from a peer GPU that already hosts the stage; 2) otherwise it loads the stage replica from the node’s pinned, shared CPU replica. Transfers use a block-wise streaming transfer to remain OOM-safe. These actions are asynchronous and only occur when a plan truly needs them.

\section{Generation of Plans}
\label{sec:planner}
In this section, we describe how $\mathcal{P}$ and $\Gamma$ are derived.

\subsection{Placement Plan}
\label{subsec: placement plan}

\emph{Dynamic Orchestrator} generates $\mathcal{P}$ following the two principles:

\underline{\textit{Principle 1:}} Eliminate avoidable inter-stage communication. To eliminate unnecessary inter-stage communication, we should, whenever possible, try to execute as many stages of each request on the same GPU; if the placement ensures minimal communication per request, the global communication is minimized.

\underline{\textit{Principle 2:}} Balance the processing speeds of the three stages to avoid stage congestion. To balance processing speeds across stages, we estimate each stage’s processing rate and allocate resources to each stage in proportion to these rates.

Guided by these principles, we later (i) further formalize \emph{placement}, (ii) derive the \emph{single-request solution} that minimizes communication for an individual request, and (iii) lift it to an \emph{overall solution} by aligning per-request choices with observed stage processing speeds.

\textbf{Definition of \emph{Placement} and \emph{Placement Type}.}
We have defined the \emph{placement} on $g$ as $\pi_{g}$, with six types:
$\pi \in$ $\{\langle E\,D\,C\rangle$, $\langle D\,C\rangle$, $\langle E\,D\rangle$, $\langle D\rangle$, $\langle E\rangle$, $\langle C\rangle\}$\footnote{We omit \(\langle E\,C\rangle\) since \(D\) dominates and lies on the critical path, so co-locating \(E\) with \(C\) neither improves throughput nor reduces \(D\)-bound traffic.}.
Since the \emph{Diffuse} stage is central and the dominant bottleneck (\S\ref{sec: preliminary}), placements containing $D$ are pivotal. We call these four types \emph{Primary Placements} $\in \{\langle E\,D\,C\rangle,\langle D\,C\rangle,\langle E\,D\rangle,\langle D\rangle\}$, and define the GPU hosting one as \emph{Primary Replica} (PR). Placements excluding $D$ are \emph{Auxiliary Placements} $\in \{\langle E\rangle,\langle C\rangle\}$, hosted on \emph{Auxiliary Replicas} (AR).
To execute a request, we select a set of GPUs whose resident stages jointly cover $\{E,D,C\}$. We call such a GPU set a \emph{Virtual Replica} (VR). As shown in Table~\ref{tbl:placement types}, VRs also have four types, in one-to-one correspondence with PR.

\begin{table}[!t]
\centering
\small
\caption{\small{Mapping between \emph{Primary Replicas}, \emph{Auxiliary Replicas}, and \emph{Virtual Replicas}}}
\label{tbl:placement types}
\resizebox{\columnwidth}{!}{%
\begin{tabular}{l l l l}
\hline
\textbf{Virtual Replica} & \textbf{Primary Replica} & \textbf{Auxiliary  Replica(s)} & \textbf{Comm. }\\
\hline
V0: $\langle E\,D\,C\rangle$                    & P0: $\langle E\,D\,C\rangle$ & (none) & 0\\
V1: $\langle D\,C\rangle + \langle E\rangle$    & P1: $\langle D\,C\rangle$    & A0: $\langle E\rangle$ & $Q_{ED}$ \\
V2: $\langle E\,D\rangle + \langle C\rangle$    & P2: $\langle E\,D\rangle$    & A1: $\langle C\rangle$ & $Q_{DC}$\\
V3: $\langle D\rangle + \langle E\rangle + \langle C\rangle$ & P3: $\langle D\rangle$ & A0: $\langle E\rangle$; A1: $\langle C\rangle$ & $Q_{ED}$+$Q_{DC}$\\
\hline
\end{tabular}%
}
\end{table}

\textbf{Single-Request Solution.}
For a given request \(r\), inter-stage communication is minimized by choosing the \emph{Virtual Replica} type that yields the smallest inter-stage transfer.
Since \(l^{C}_{\text{proc}} \textgreater l^{E}_{\text{proc}}\) and communication \(Q \propto l\), we have \(Q_{DC} \textgreater Q_{ED}\); thus, as shown in Table~\ref{tbl:placement types}, the communication of a \emph{Virtual Replica} type increases monotonically with its index.
Hence, for each request \(r\), select the first \emph{feasible} \emph{Virtual Replica} type in the order \(V0 \prec V1 \prec V2 \prec V3\).
We denote this minimal-communication choice by \(\operatorname{OptVR}(r)\).

\textbf{Overall Solution.}
When each request \(r\) runs on its \(\operatorname{OptVR}(r)\), total inter-stage communication is minimized.
To realize this under a finite GPU set \(\mathcal{G}\), the \emph{placement plan} should provision \emph{Virtual Replica} types in proportions that mirror the distribution of \(\operatorname{OptVR}(r)\) observed in \(\mathcal{R}\).

Concurrently, to prevent stage imbalances that lead to congestion, the \emph{Primary} and \emph{Auxiliary Replicas} within each \emph{Virtual Replica} type should be apportioned inversely to their respective processing speeds.

Following this principle, we obtain the placement allocation algorithm in Algorithm~\ref{alg:waterfall}.

\begin{algorithm}[!t]
\caption{Generation of Placement Plan $\mathcal{P}$}
\label{alg:waterfall}
\begin{algorithmic}[1]
\Require request set \(\mathcal{R}\) with per-request Peak Memory func $\mathrm{peakMem}(\cdot)$;
\emph{Virtual Replica} Type set $\mathcal{T}$ with per-type residual memory func $\mathrm{cap}(\cdot)$;
GPU numbers \(G\); processing speed set of different placement type \(\{v_{\pi}\}\).
\Ensure $\mathcal{P}=\{\small \pi_{g} \mid  g \in [0,G-1]\}$
\For{$r \in \mathcal{R}$}                                          \label{line:optvr-begin}
  \State \(\operatorname{OptVR}(r) \gets \min \{\, t \in \mathcal{T} \mid \mathrm{peakMem}(t) \le \mathrm{cap}(t) \,\}\)
\EndFor                                                            \label{line:optvr-end}
\For{$t \in \mathcal{T}$}                                          \label{line:mix-begin}
  \State \(\alpha_t \gets \frac{|\{\, r \in \mathcal{R} : \operatorname{OptVR}(r)=t \,\}|}{|\mathcal{R}|}\)
  \quad \(N_t \gets \left\lfloor \alpha_t\, G\right\rfloor\)
\EndFor                                                            \label{line:mix-end}
\For{$t \in \mathcal{T}$}                                          \label{line:split-begin}
  \State \((N_t^{\mathrm{prim}},\,N_t^{\mathrm{aux}}) \gets \mathrm{Split}(N_t,\{v_{\pi}\}, t)\)
\EndFor                                                            \label{line:split-end}
\State $\mathcal{P} \gets \mathrm{PackPerMachine}\!\left(\{(t,N_t^{\mathrm{prim}},N_t^{\mathrm{aux}})\}_{t\in\mathcal{T}},\, G\right)$ \label{line:place}
\end{algorithmic}
\end{algorithm}

Lines~\ref{line:optvr-begin}--\ref{line:optvr-end} scan requests and pick the \(\operatorname{OptVR}(r)\) based on their peak memory and the capacity of each \emph{Virtual Replica} type.
Lines~\ref{line:mix-begin}--\ref{line:mix-end} determine the number $N_t$ of each type of \emph{Virtual Replica} occupying the GPU according to the proportion of $\operatorname{OptVR}(r)$.
Lines~\ref{line:split-begin}--\ref{line:split-end} split each \(N_t\) into \(\bigl(N_t^{\mathrm{prim}}, N_t^{\mathrm{aux}}\bigr)\) using the monitored throughputs of the corresponding \emph{Primary} and \emph{Auxiliary} placement types in \(\{v_\pi\}\) to balance service rates within each type.
Line~\ref{line:place} generates $\mathcal{P}$ based on the results and prioritizes assigning identical \(\pi\) types within each machine. Details of {\small \texttt{Split()}} and {\small \texttt{PackPerMachine()}} is shown in Appendix~\ref{subsec: placement plan detail}

\subsection{Dispatch Plan}
\label{subsec: dispatch plan}
\textsc{TridentServe} employs a \emph{Resource-Aware Dispatcher} to produce \emph{dispatch plans} whose objective is to maximize the SLO-attainment, as SLO compliance is paramount in GVT serving~\cite{diffserve,jones2016slo}.

\textbf{Discussion.}
We formalize the ideal objective as follows: for each request \(r\) with SLO deadline \(d_r\), and a cluster of \(G\) GPUs with placements \(\{\pi_g\}_{g=0}^{G-1}\), select, for every stage \(s\), the dispatch plan \(\Gamma_r^s\) and its dispatch time \(t_{\Gamma_r^s}\) so as to maximize the number of requests that finish within their SLOs:
\[
  \max \sum_{r} \mathbf{1}\!\Big[T_r\big(\{\Gamma_r^s,t_{\Gamma_r^s}\}_s\big) \le d_r\Big],
\]
As shown in Appendix~\ref{sec:optimal scheduling}, this is a strengthened NP-complete \emph{Job-Shop Scheduling Problem}~\cite{jobshop} and is intractable for real-time decision making, rendering it impractical for online serving. Accordingly, we \emph{dissect} the problem into tractable parts and present a \emph{two-step} solution. 
We first further define the \emph{dispatch plan} and then detail the two-step solution.

\textbf{Definition of \emph{Dispatch Plan}.}
A \emph{dispatch plan} for request $r$'s stage $s$ is defined as: 
$
  \Gamma^s_r
  \;=\;
  \bigl(
      r,\;
      \mathcal{G}^s_r,\;
      \{\,s:\phi_s \}
  \bigr),
$
where $r \in \mathcal{R}$ is the dispatched request,  $\mathcal{G}^s_r\subseteq\mathcal{G}$ is the GPU worker set that this plan assigned to, $s$ is the stage that will execute, $\phi_s$ is the parallel config of $s$.
For example,
$
  \Gamma^D_2 \;=\;
  \bigl(
      r_2,\;
      \{g_0,g_1,g_2,g_3\},\;
      \{\,D:\text{SP}_4\,\}
  \bigr)
$
means send request \#2 to GPUs \#0–3, then run \textit{Diffuse} stage with 4-degree sequence parallel to the request.

A request $r$'s \emph{dispatch plan} is the set of all dispatch plans it participates in:
$
  \Gamma_r \;=\; \{\, \Gamma^{E}_r,\; \Gamma^{D}_r,\;\; \Gamma^{C}_r \,\},
$

\textbf{Two-step Formulation.}
We dissect this long-range continuous time scheduling problem into a two-step decision per tick: 1) for each pending request \(r\), we decide whether to dispatch it \emph{now} and first determine its \emph{Diffuse} stage plan \(\Gamma^D_r\), 2) then we derive \(\Gamma^E_r\) and \(\Gamma^C_r\) from \(\Gamma^D_r\).
We simplify the problem based on the following two insights:

\underline{\textit{Insight 1:}} Because the pending requests evolve continuously and runtime estimates are jitter-prone, long-horizon decisions are brittle and slow~\cite{Ousterhout2013Sparrow}. A per-tick, myopic decision yields responsive and robust online scheduling~\cite{Dean2013TailAtScale,HarcholBalter2013PMDCS}.

\underline{\textit{Insight 2:}} As shown in Figure~\ref{fig:time breakdown}, the \emph{Diffuse} stage dominates end-to-end latency. Meanwhile, \emph{Diffuse} is highly sensitive to parallel degree, whereas \emph{Encode} and \emph{Decode} stages contribute little to latency and are nearly insensitive to parallelism (\S\ref{sec:analysis}). 
Hence, we delicately decide \(\Gamma^D_r\) first. Then, \(\Gamma^E_r\) and \(\Gamma^C_r\) are then set directly on profiled \emph{optimal parallelism strategy}\footnote{We define the \emph{optimal parallelism strategy} in a practical way: the highest degree parallelism whose \emph{efficiency > 0.8}. Efficiency = $\frac{\text{Actual \ Speedup}}{\text{Theoretical \ Speedup}}$.} and the \emph{Primary Placement} type where $\mathcal{G}^D_r$ assigned to: if \(E\) co-resides with \(D\), reuse the \(\mathcal{G}^D_r\) as \(\mathcal{G}^E_r\) since $E$ takes almost no time, it can be directly merged and executed with D in \emph{Runtime Engine}; if \(C\) co-resides with \(D\), choose a subset \(\mathcal{G}^C_r \subseteq \mathcal{G}^D_r\) since \emph{Decode} usually needs less resource than \emph{Diffuse} stage (\S\ref{sec:analysis}), which can avoid inter-stage communication; otherwise, assign \(E\) and/or \(C\) to idle \emph{Auxiliary Replicas} using the profiled optimal parallelism strategies.

We therefore first solve for \(\Gamma^D_r\) and subsequently instantiate \(\Gamma^E_r\) and \(\Gamma^C_r\) based on \(\Gamma^D_r\).

\begin{figure}[!t]
    \centering
    \includegraphics[width=\linewidth]{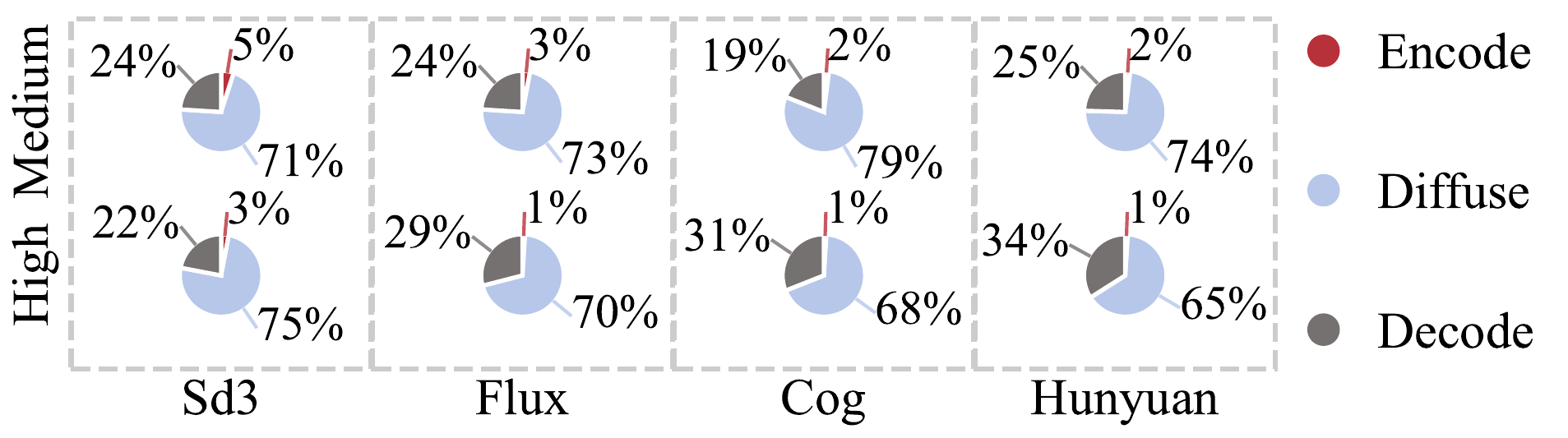}
    \caption{\small{Time breakdown of different models tested on NVIDIA L20. High and Medium are the workloads that will be introduced in \S\ref{subsec: experiment setup}.}}
    \label{fig:time breakdown}
\end{figure}

\textbf{Solution for $\Gamma^D$.} We first formulate an Integer Linear Programming~\cite{ilp} (ILP) to solve the current dispatched $\Gamma^D_r$:

\medskip
\noindent\textit{\underline{Sets and indices.}}
Let \(r \in \mathcal{R}\) be request;\;
\(i \in \mathcal{I}=\{0,1,2,3\}\) be Primary types;
\(k \in \mathcal{K}=\{1,2,4,8\}\) be parallel degrees.

\medskip
\noindent\textit{\underline{Parameters.}}
Let \(d_r\) be the SLO deadline of request \(r\);\;
\(\tau\) be the current scheduling time;\;
\(t_{r,i,k}>0\) be the pre-profiled runtime of \(r\), which runs the stages residing on type \(i\) with degree \(k\);\;
\(B_{i}\in\mathbb{Z}_{\ge 0}\) be the number of idle \emph{Primary replicas} of type \(i\);
\(Q_{r,i}\ge 0\) be the inter-stage communication penalty weight if type \(i\) is chosen;\;
\(W_{r}>0\) be the SLO-aware reward weight;\;
\(M\) be a sufficiently large constant.

\medskip
\noindent\textit{\underline{Feasible Indicator.}}
This indicator can be calculated in advance to filter and reduce the number of variables:
\(E_{r,k}\in\{0,1\}\) filters the inefficient\footnote{We set the inefficient threshold as 0.8 since it is widely used~\cite{Pacheco2011IntroParallelProgramming,Grama2003ParallelComputing}.} degree \(k\) for \(r\);\;
\(F_{r,i,k}\in\{0,1\}\) filters the unfeasible \emph{Primary Placement} type \(i\) for \(r\);\;

\medskip
\noindent\textit{\underline{Decision variables.}}
\(x_{r,i,k}\in\{0,1\}\) means if dispatch \(r\) now on \emph{Primary Replica} of type \(i\) with degree \(k\);\;
\(D_r\in\{0,1\}\) means if \(r\) is finished on time.

\medskip
\noindent\emph{\underline{Objective:}}
\begingroup\small
\[
\max_{x,D}\ \sum_{r\in\mathcal{R}}\sum_{i\in\mathcal{I}}
\sum_{k\in\mathcal{K}}
\bigl(W_{r}-Q_{r,i}\bigr)\,x_{r,i,k}
\tag{OBJ}
\]

\noindent\emph{\underline{Constraints:}}
\begin{align}
x_{r,i,k} \le E_{r,k}F_{r,i,k} && \forall\, r,i,k, && \text{(C0)} \notag \\[2pt]
\sum_{i,k} x_{r,i,k} \le 1 && \forall\, r, && \text{(C1)} \notag \\[2pt]
\sum_{r\in\mathcal{R}}\sum_{k\in\mathcal{K}} k\, x_{r,i,k} \le B_{i} && \forall\, i\in\mathcal{I}, && \text{(C2)} \notag \\[2pt]
\tau + \sum_{i,k} t_{r,i,k}\, x_{r,i,k} \le d_r + M(1-D_r) && \forall\, r, && \text{(C3a)} \notag \\[-2pt]
D_r \le \sum_{i\in\mathcal{I}}\sum_{k\in\mathcal{K}}x_{r,i,k},\quad D_r \in \{0,1\} && \forall\, r, && \text{(C3b)} \notag \\[2pt]
x_{r,i,k} \in \{0,1\} && \forall\, r,i,k. && \text{(C4)} \notag
\end{align}
\endgroup

\textbf{(OBJ)} maximizes an SLO-aware reward while penalizing inter-stage communication: \(W_{r}\) is a pre-defined weight that favors on-time completion and discourages lateness. We also add an aging mechanism~\cite{Silberschatz2018OSC} to avoid starvation in the setting of \(W_{r}\), detailed in Appendix~\ref{subsec: dispatch plan details}; \(Q_{r,i}\) is the pre-defined communication penalty that encourages choosing the \emph{Primary Replica} with less communication, also detailed in Appendix~\ref{subsec: dispatch plan details}. 

\textbf{(C0)} prunes inefficient and infeasible degree/type pairs for each request.
\textbf{(C1)} permits at most one assignment per request.
\textbf{(C2)} enforces that the assigned \emph{Primary Replicas} do not exceed their total amount.
\textbf{(C3a)} links the chosen runtime to the SLO using a big-\(M\) constraint with current time \(\tau\); \textbf{(C3b)} allows \(D_r=1\) only if the batch is dispatched.

Due to the filtering of preliminary variables and the fact that each cluster generally has 1-2 \emph{Primary Placement} types (\S\ref{subsec: case study}), this method can achieve a solution in under a hundred milliseconds for larger-scale clusters, as detailed in \S\ref{subsec: Sensitivity and Scalability}.

After solving this efficient ILP, we obtain an assignment \(x_{r,i,k}=1\) for each dispatched request \(r\), and select an \emph{intra-machine} GPU set \(\mathcal{G}^D_r\) that provides \(k\) \emph{Primary Replicas} of type \(i\) (to avoid cross-machine, if not found, stay undispatched for next round).
From the profiler, we retrieve the best parallelism strategy \(\phi_D\) for the \emph{Diffuse} stage at degree \(k\), and form
\[
  \Gamma^{D}_r \;=\; \bigl(\ r,\ \mathcal{G}^D_r,\ \{\,D:\phi_D\,\}\ \bigr).
\]
Subsequent \(\Gamma^{E}_r\) and \(\Gamma^{C}_r\) follow directly from \(\Gamma^{D}_r\) per the earlier insight.

\textbf{Solution for \(\Gamma^E\) and \(\Gamma^C\).}
Given \(\Gamma^D_r\) and the assigned \emph{Primary Replica} type:
\begin{itemize}
  \item \(\Gamma^E_r\). Select the profiled \emph{optimal parallelism} \(\phi_E\). If the \emph{Primary Replica} contains \(E\), reuse \(\mathcal{G}^D_r\); otherwise, choose an idle or earliest-to-finish—GPU set (reported by the \emph{Monitor}) \(\mathcal{G}^E_r\) from \(E\)-type \emph{Auxiliary Replicas}.
  \item \(\Gamma^C_r\). Select the profiled \emph{optimal parallelism} \(\phi_C\). If the \emph{Primary Replica} contains \(C\), take a subset \(\mathcal{G}^C_r \subseteq \mathcal{G}^D_r\); otherwise, choose an idle or earliest-to-finish—GPU set (reported by the \emph{Monitor}) \(\mathcal{G}^C_r\) from \(C\)-type \emph{Auxiliary Replicas}.
\end{itemize}

\section{Implementation}
\label{sec:implementation}
The core of \textsc{TridentServe} comprises about 12K lines of Python/Triton~\cite{tillet2019triton} code in total: the \emph{Runtime Engine} is ~10K LOC and the Planners are ~2K LOC.
We have integrated popular diffusion models (roughly 16K LOC in total), including \emph{Stable-Diffusion}~\cite{stablediffusion3}, \emph{PixArt}~\cite{chen2024pixartdelta,chen2023pixartalpha}, \emph{Flux}~\cite{flux}, \emph{Cogvideo}~\cite{cogvideox}, \emph{HunyuanVideo}~\cite{hunyuanvideo}, and \emph{HunyuanDiT}~\cite{hunyuandit}.
We also support a dynamic batching mechanism when there are multiple lightweight requests to be served.
Besides, although our solution above concentrates on \emph{SP}, it is compatible with \emph{MP} by simply treating multiple devices as one. 
We leave more details of dynamic batching and \emph{MP} integration in Appendix~\ref{sec: extend method}.
In the \emph{Runtime Engine}, we implement asynchronous concurrent execution with
\emph{ray}~\cite{ray} and coroutines~\cite{coroutine}, and use \emph{NCCL}~\cite{NCCL} and \emph{NIXL}~\cite{nixl} as communication backends.
To implement the \emph{Dispatch plan} solver, we use the PuLP~\cite{pulp} libraries for solving the ILP problem.

\section{Evaluation}
\label{subsec: experiment}
\subsection{Experimental Setup}
\label{subsec: experiment setup}
\textbf{Testbed.}
We conduct our experiments on a cluster of 16 GPU servers, each equipped with \(8\times\) NVIDIA L20~\cite{l20} (48G), for a total of 128 GPUs. 
Within each server, GPUs connect to the host via PCIe~4.0~x16 and are mapped into a 4+4 dual-NUMA topology.
The servers are interconnected using Mellanox ConnectX\mbox{-}5 adapters operating at 100\,Gb/s Ethernet. 

\textbf{Models.}
To demonstrate the generality of our approach, we evaluate four widely used generative vision pipelines: \emph{StableDiffusion3-Medium}~\cite{sd3model} (Sd3) and \emph{Flux.1}~\cite{fluxmodel} (Flux) for image generation, \emph{CogVideoX1.5-5B}~\cite{cogvideoxmodel} (Cog), and \emph{HunyuanVideo}~\cite{hunyuanvideomodel} (Hunyuan) for video generation. 
In particular, \emph{Sd3} and \emph{Cog} can be deployed in a fully co-located manner, whereas \emph{Flux} and \emph{Hunyuan} require disaggregated deployment to avoid OOM, thereby demonstrating that our method automatically adapts to diverse deployment scenarios.
Their concrete configurations have been summarized in Table~\ref{tab:diffusion pipelines}. 

\textbf{Workloads.}
We use three classes of workloads:
(1) \emph{Steady} workloads were generated using a fixed arrival rate within 30 min. We synthesize three \emph{Steady} traces: \emph{light}, \emph{medium}, and \emph{heavy}, with a bundle of mixed target resolution and durations.
(2) The \emph{Dynamic} workload randomly interleave the three \emph{Steady} workloads above while varying the arrival proportion following the pattern in Figure~\ref{fig:workload pattern}.
(3) The \emph{Proprietary} workload is proportionally scaled to adapt to our cluster and model from real production traces whose pattern is shown in Figure~\ref{fig:workload pattern}, which exhibits pronounced diurnal/tidal effects.
We leave the details of traces and inference settings in Appendix~\ref{subsec: workload setting}.
At the same time, referring to Alpaserve's~\cite{alpaserve} approach, we set the SLO to 2.5 $\times$ the latency when it is running under the \emph{optimal parallelism strategy}.

\begin{figure}[!t]
    \centering
    \includegraphics[width=0.95\linewidth]{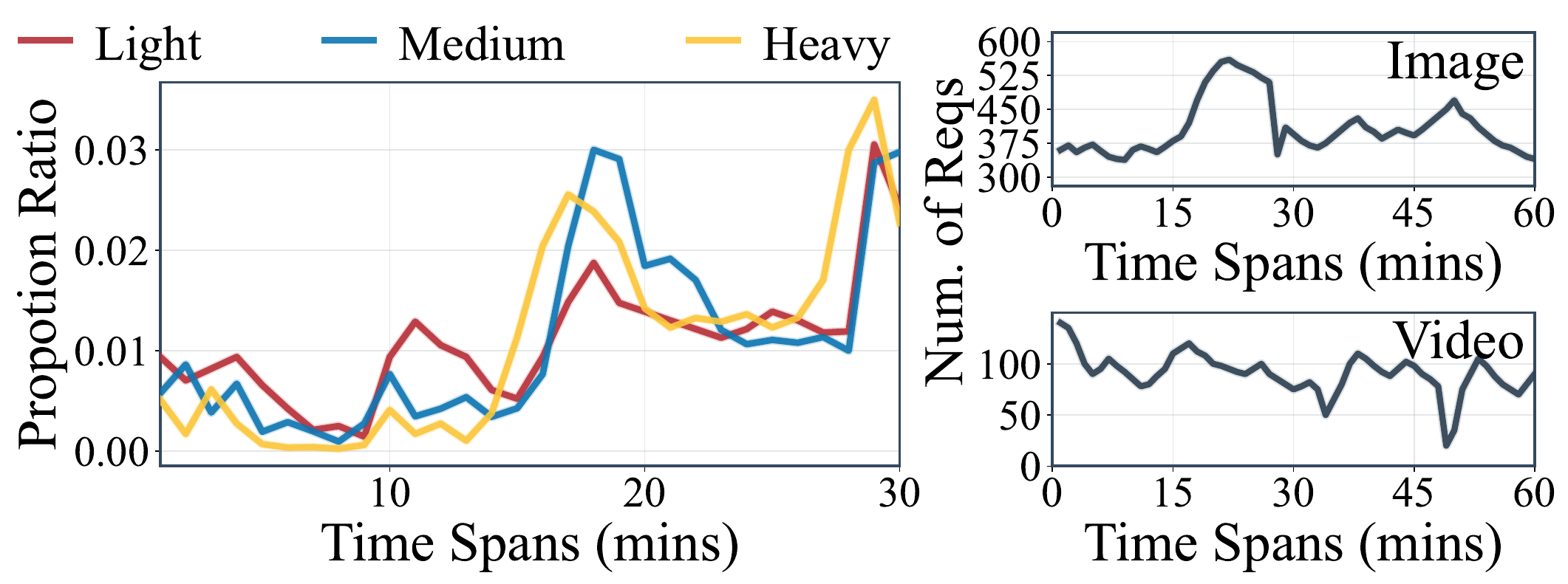}
    \caption{\small{Left: the pattern of the \emph{Dynamic} workload, which is to distribute three different steady workloads according to the proportions in the diagram within a Time Span.
    Right: segments of proprietary traces for image and video generation.}}
    \label{fig:workload pattern}
\end{figure}

\begin{figure*}[t]
    \centering
    \includegraphics[width=\linewidth]{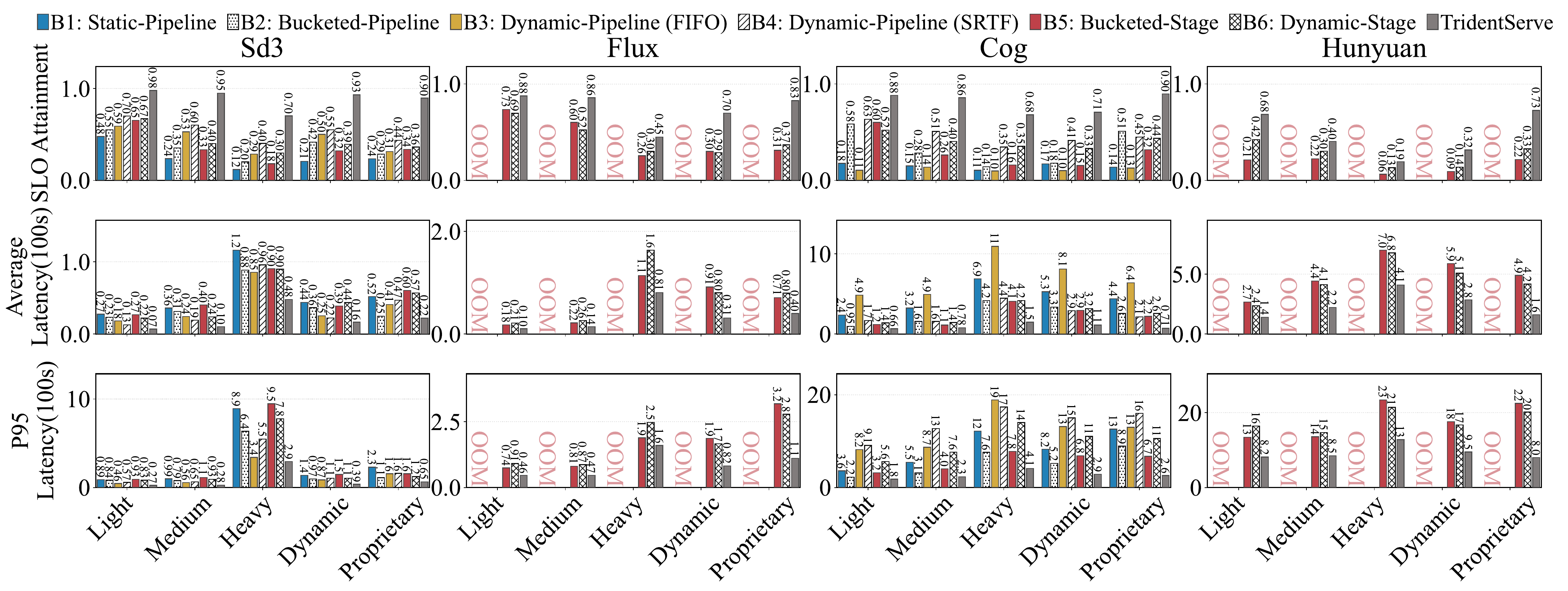}
    \caption{\small{End-to-End Results.
    }}
    \label{fig:main_results}
\end{figure*}

\textbf{Baselines.} 
\label{sec:baselines}
To the best of our knowledge, all existing works
can only employ static, pipeline-level resource allocation (\textbf{B1}). To comprehensively assess the effectiveness of our work, we further consider 5 variants of TridentServe (\textbf{B2-B6}), constituting 6 baselines in total.
Further implementation details are deferred to Appendix~\ref{subsec: baseline implementation}.

\begin{description}
  \item[\textit{B1: Static Pipeline-level}]
  Co-locate all stages, static parallelism strategy (degree k is chosen to satisfy the SLO at the maximum length) for all requests, and use the same resources across stages within one request. FIFO scheduling. This is the scheme used by xDiT~\cite{xdit}.

 \item[\textit{B2: Bucketed Pipeline-level}]
    Co-locate all stages and use the same resources across stages within one request. We statically partition the cluster into disjoint buckets by degree \(k\); each request is routed to the bucket matching its optimal degree of the \emph{Diffuse} stage from the profiler. Bucket capacities are sized in proportion to demand and per-instance service rate. FIFO within each bucket.
  \item[\textit{B3: Dynamic Pipeline-level (FIFO)}]
  Co-locate all stages. Per-request dynamic parallelism: upon arrival, use the \emph{optimal parallelism strategy} of the \emph{Diffuse} stage, and all three stages use the same resource. FIFO scheduling.
  \item[\textit{B4: Dynamic Pipeline-level (SRTF)}]
  As in B3, but scheduled by \emph{Shortest Remain Time First}~\cite{srtf} (with aging).
  \item[\textit{B5: Bucketed Stage-level}]
  Manual disaggregated stages deployment to enable stage-level resource allocation. Every disaggregated cluster uses the bucketed method in \textit{B2}. FIFO scheduling.

  \item[\textit{B6: Dynamic Stage-level (SRTF)}]
  Manual disaggregated stages deployment to enable stage-level resource allocation. Each stage selects its \emph{optimal parallelism strategy} dynamically from the profiler. SRTF scheduling.
\end{description}

\subsection{End-to-End Evaluation}
\label{subsec:e2e}
We evaluate \textsc{TridentServe} against the baselines in \S\ref{sec:baselines} on all workloads, reporting SLO attainment, mean latency, and P95 latency (Figure~\ref{fig:main_results}).

Overall, \textsc{TridentServe} increases SLO attainment and reduces latency while eliminating OOMs: all \textbf{B1--B4} runs OOM on \textit{Flux} and \textit{HunyuanVideo} across all workloads, whereas ours never OOMs via automatic dynamic deployment. Compared to \textbf{B1} (xDiT), we observe 2.1$\times$-6.6$\times$ higher SLO attainment, mean-latency speedups of 2.4$\times$-6.2$\times$, and P95-latency speedups of 2.0$\times$-4.8$\times$. Versus other colocated and pipeline-level methods (\textbf{B2--B4}), we improve SLO by 1.4$\times$-8.3$\times$, and reduce mean and P95 latency by 1.1$\times$-9.0$\times$ and 1.2$\times$-6.1$\times$, respectively. Against stage-level baselines, we achieve 1.2$\times$-4.6$\times$ higher SLO attainment, 1.4$\times$-4.0$\times$ and 1.2$\times$-4.1$\times$ lower mean and P95 latency than \textbf{B5}, and 1.3$\times$-2.5$\times$ higher SLO attainment, 1.7$\times$-3.6$\times$ and 1.5$\times$-4.1$\times$ lower mean and P95 latency than \textbf{B6}. We detail the observations below.

\textbf{Static Pipeline-level Baselines (B1 \& B2).}
\textbf{B1} uses a single global parallelism strategy for all requests, which mismatches heterogeneous demand (\S\ref{sec:analysis}), lowering efficiency and performance.
\textbf{B2} improves on B1 by routing each request to a fixed \(k\)\,bucket, but arrival randomness creates structural imbalance across buckets—some are idle while others queue. The imbalance is mild on steady traces but pronounced on dynamic ones. Our method employs stage-level parallelism and accounts for cluster utilization, outperforming B1 \& B2 via higher resource efficiency.

\textbf{Dynamic Pipeline-level Baselines (B3 \& B4).}
\textbf{B3} uses a per-request dynamic parallelism yet schedules with FIFO, which induces head-of-line blocking and under-utilization.
\textbf{B4} switches to SRTF, which improves utilization and prioritizes near-deadline requests, thus improving performance.
However, it still uses a fixed \emph{optimal parallelism strategy}, which lacks flexibility and therefore limits its upper bound.

\begin{figure}[!t]
    \centering
    \includegraphics[width=1\linewidth]{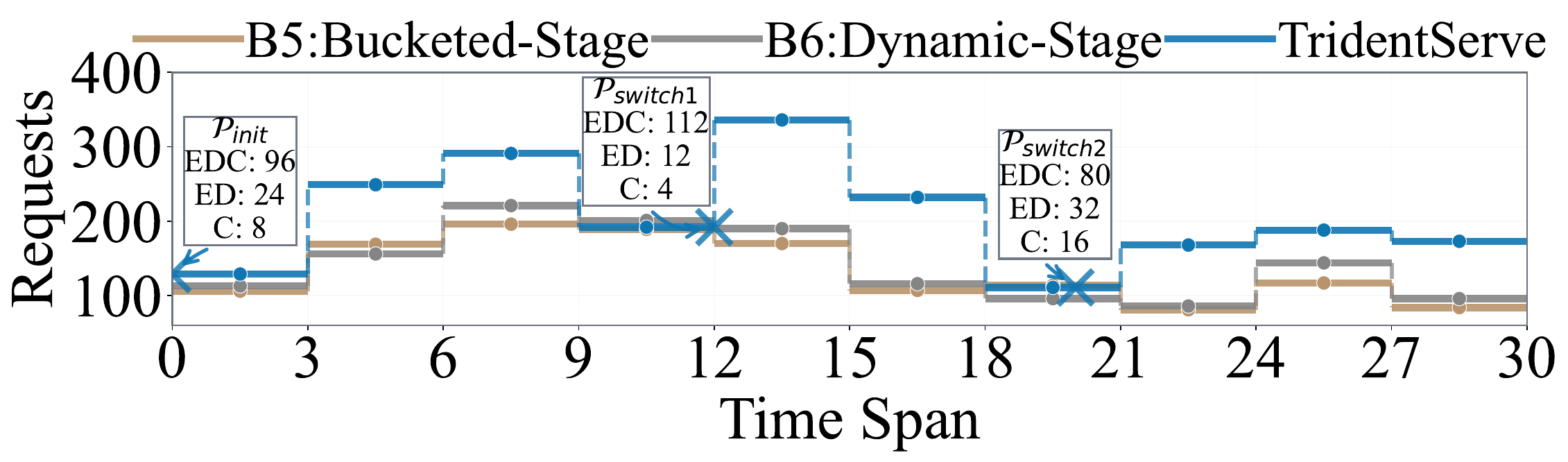}
    \caption{\small{The throughput per time span and placement switch of \emph{Flux} on \emph{Dynamic} workload.}}
    \label{fig:dynamic_workload}
\end{figure}

\textbf{Stage-level Baseline (B5 \& B6).}
\textbf{B5} and \textbf{B6} adopt disaggregated deployment for stage-level allocation, effectively eliminating OOMs. Yet their manually designed, \emph{static} placements introduce extra inter-stage transfers and drift out of alignment as arrival patterns change. For relatively stable workloads, our method yields, on average, a 2.1$\times$ increase in SLO satisfaction and 2.1$\times$/2.3$\times$ reductions in mean/P95 latency over \textbf{B5} and \textbf{B6}. Under more dynamic loads—e.g., \emph{Dynamic} and \emph{Proprietary}—the gains are larger: 2.7$\times$ higher SLO satisfaction and 2.6$\times$/2.6$\times$ faster mean/P95 latency, exceeding the improvements on stable traces.

\subsection{Case Study}
\label{subsec: case study}
We analyze the distribution of \emph{Virtual Replicas} (VR) used by requests and placement shifts under the \emph{Dynamic} 
workload.

\textbf{Virtual Replica Statistics}
Figure~\ref{fig:virtual replica} reports the VR distribution for \emph{Flux} and \emph{HunyuanVideo}. Most requests run on V0 (least inter-stage communication); requests that exceed V0 memory are placed on V1 or V2. In \emph{Flux}, 84\% of requests are V0-eligible, and our scheduler dispatches 80\% to V0; in \emph{HunyuanVideo}, 87\% are eligible, and we dispatch 84\%. Thus, the method steers nearly all requests to the lowest-communication plan. 
The small remainder reflects transient congestion of the corresponding optimal VR during scheduling, which will soon be eliminated by our placement switch.

\textbf{Placement Switching under Dynamic Workloads.}
Taking \emph{Flux} as an example, Figure~\ref{fig:dynamic_workload} visualizes placement switching under the \emph{Dynamic} workload. At around 12\,min, as many \emph{Light} requests arrive, our \emph{orchestrator} switches more placements to \emph{EDC} to accommodate the surge, recovering throughput during a downturn; around 20\,min, when a large number of \emph{Heavy} requests arrive, we increase the number of \emph{ED} placements, preventing throughput from continuing to decline. In contrast, \textbf{B5} and \textbf{B6} use static placements, which lead to a sustained decrease in throughput.

\begin{figure}[H]
    \centering
    \begin{minipage}[b]{0.48\linewidth}
        \centering
        \includegraphics[width=\linewidth]{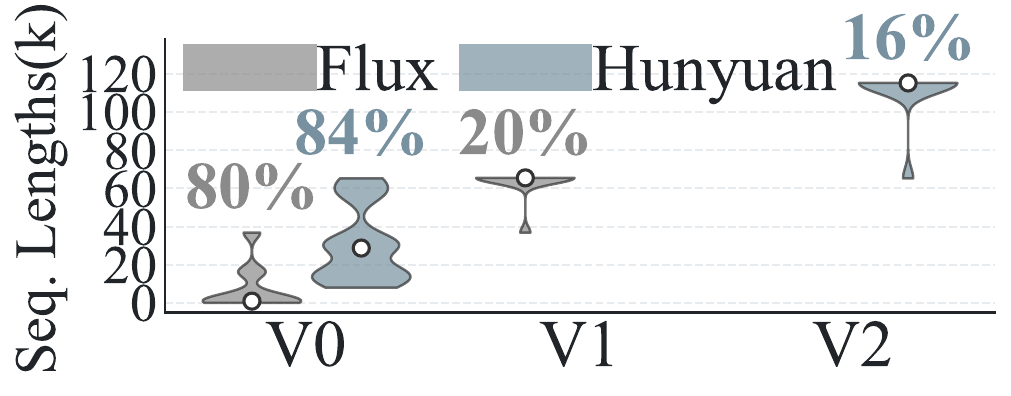}
        \caption{\small{Distribution of Virtual Replica types. The white circle indicates the median.}}
        \label{fig:virtual replica}
    \end{minipage}
    \hfill
    \begin{minipage}[b]{0.48\linewidth}
        \centering
        \includegraphics[width=\linewidth]{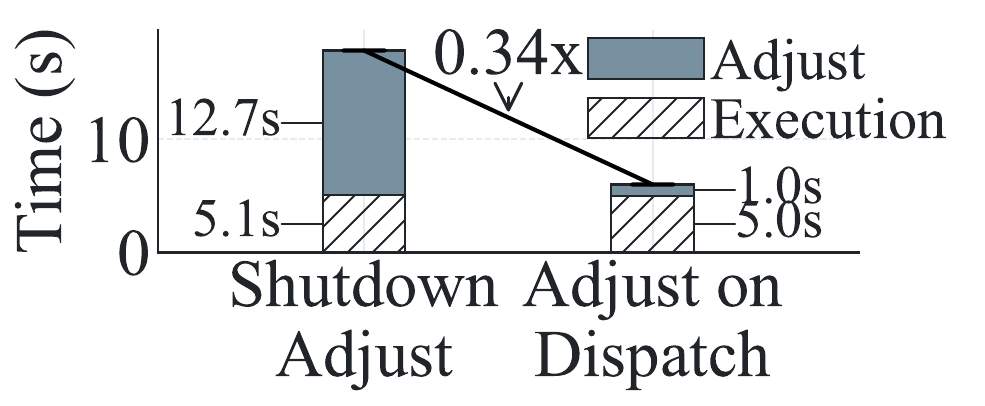}
        \caption{\small{Comparison between shutdown adjust and \emph{Adjust-on-Dispatch}.}}
        \label{fig:adjust on dispatch}
    \end{minipage}
\end{figure}

\subsection{Ablation Study}
\label{subsec:ablation}
\textbf{E2E Ablation} We conduct ablations on \emph{HunyuanVideo} and \emph{Flux} using the \emph{dynamic} and \emph{steady (medium)} workloads, removing one component at a time:
\begin{itemize}
  \item \textbf{wo-switch}: disable placement switch, only use $\mathcal{P}_{init}$.
  \item \textbf{wo-stageAware}: disable stage-level resource allocation, and align all GPU demand with the \emph{Diffuse} stage.
  \item \textbf{wo-scheduler}: replace our \emph{Resource-Aware Dispatcher} with a simple greedy SRTF policy; each stage still uses profiled-\emph{optimal parallelism strategy}.
\end{itemize}
Results (Figure~\ref{fig:ablation}) show:
(i) \emph{placement switching} is crucial under dynamic load shifts, reducing latency by 33.1\% and 17.6\% in \emph{Flux} and \emph{HunyuanVideo}; under the \emph{Steady} workload, the gains are more modest—11.7\% and 6.2\% respectively. It also improves SLO attainment by mitigating stage-rate imbalance.
(ii) on average, \emph{stage-level allocation} consistently reduces latency by 22.3\% and 25.3\% and improves SLO attainment by 10.0\% and 24.2\% in \emph{Flux} and \emph{HunyuanVideo}, regardless of workload, owing to more efficient resource use;
(iii) the \emph{global scheduler} substantially boosts SLO attainment up to 36.8\% and 42.4\% in \emph{Flux} and \emph{HunyuanVideo}, while also yielding a certain reduction in latency.

\textbf{Ablation for \emph{Adjust-on-Dispatch}.} To assess the effectiveness of \emph{Adjust-on-Dispatch}, we use Flux’s \emph{Dynamic} workload and monitor a certain request (1024p) that completes immediately before an adjustment is required.
We compare its completion time under two scenarios—naïve downtime adjustment and \emph{Adjust-on-Dispatch}, with the breakdown shown in Figure~\ref{fig:adjust on dispatch}. The results show that downtime adjustment incurs substantial overhead because the system must shut down for an extended interval during which no requests are processed, whereas our approach amortizes the adjustment by performing it concurrently between dispatches, adding virtually no runtime.
\begin{figure}[!t]
    \centering
    \includegraphics[width=\linewidth]{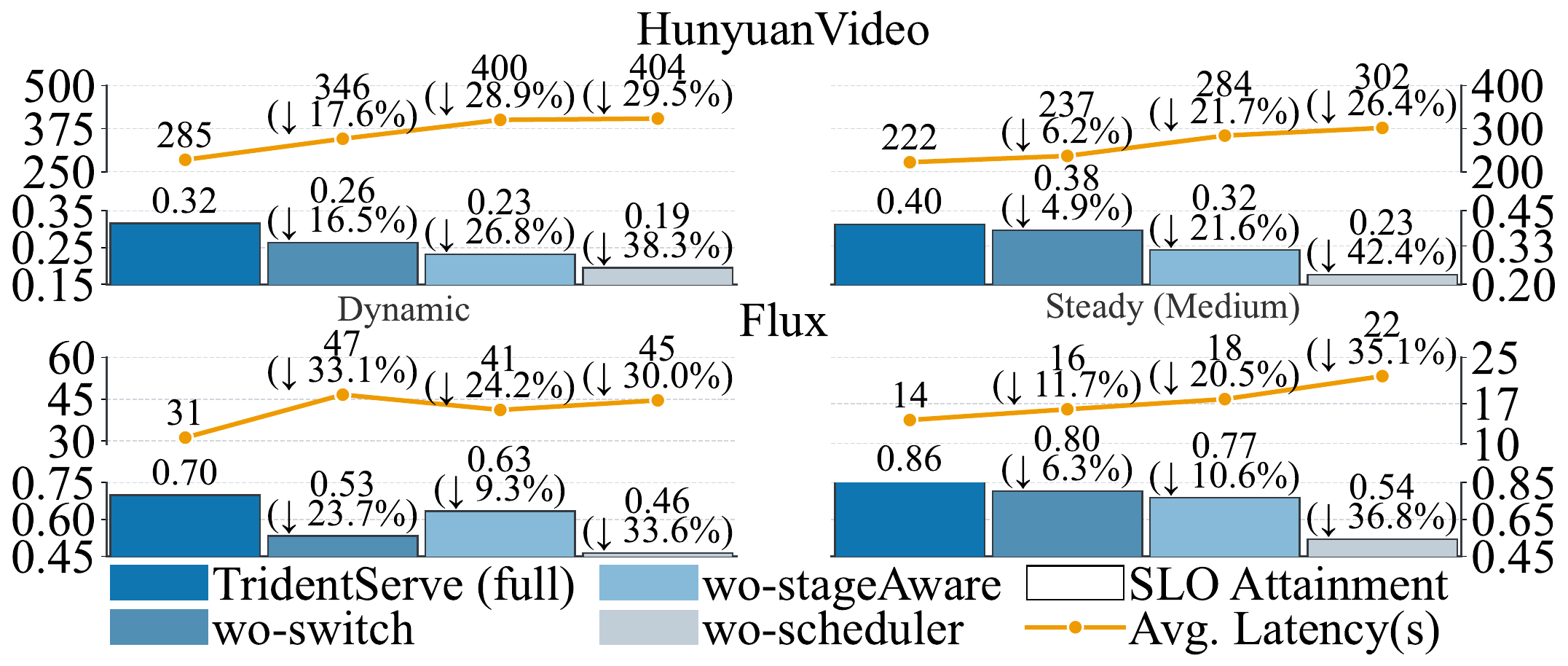}
    \caption{\small{Ablation Study for \textsc{TridentServe}.
    }}
    \label{fig:ablation}
\end{figure}
\subsection{Sensitivity and Scalability}
\label{subsec: Sensitivity and Scalability}

\textbf{SLO Sensitivity.}
To factor out the effect of SLO settings on SLO attainment, following AlpaServe~\cite{alpaserve}, we report \emph{SLO scaling} on the \emph{Dynamic} workload (Figure~\ref{fig:slo sensitivity}). \textsc{TridentServe} consistently achieves higher SLO attainment than all baselines, demonstrating robustness across varying SLO targets.

\begin{figure}[!t]
    \centering
    \includegraphics[width=\linewidth]{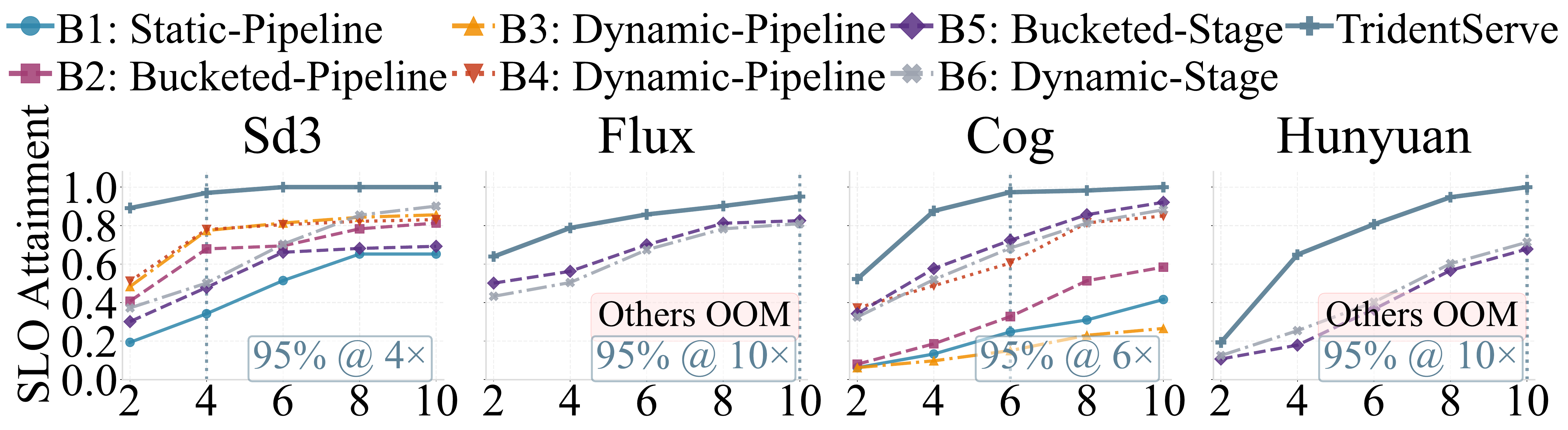}
    \caption{\small{SLO sensitivity: the x\mbox{-}axis is the scale factor $\alpha$; for each $\alpha$, we set the SLO to $\alpha \times$ the latency under the optimal parallelism.}}
    \label{fig:slo sensitivity}
\end{figure}

\textbf{\emph{Dispatcher} Scalability.}
We validate scalability by timing a single \emph{dispatcher} solve, which is the system’s most time-consuming step.
We emulate larger clusters by scaling the average number of pending requests in proportion to the GPU count, keeping the request/GPU ratio fixed and extrapolating our 128-GPU cluster to thousands of GPUs.
As shown in Table~\ref{tab:planner-scalability}, the per tick solve time remains within hundreds of milliseconds, demonstrating strong scalability.

\begin{table}[!t]
  \centering
  \setlength{\tabcolsep}{8pt}
  \caption{\small{\emph{Dispatcher} scalability: solver time per scheduling tick.}}
  \label{tab:planner-scalability}
  \small
  \begin{tabular}{@{}l *{5}{S[table-format=5.0]} @{}}
    \toprule
    \textbf{\#GPUs}            & {128} & {256} & {512} & {1024} & {4096} \\
    \midrule
    \textbf{Time (ms)}         & {25} & {26} & {36} & {45}  & {98}   \\
    \bottomrule
  \end{tabular}
\end{table}

\section{Related Work}
\textbf{Diffusion Pipeline Serving.}
With the rapid rise of \emph{Diffusion Pipelines}, several works target serving optimization. 
NIRVANA~\cite{nirvana} reuses intermediate noise across requests to reduce denoising steps; 
DIFFSERVE~\cite{diffserve} cascades small and large models to balance speed and quality; 
PATCHEDSERVE~\cite{patchedserve} caches at the patch-level to exploit diffusion redundancy; 
FlexCache~\cite{flexcache} stores and compresses intermediate states by leveraging step similarity for acceleration.
All of the above are \emph{lossy}, algorithmic accelerations and thus orthogonal to our approach.
To the best of our knowledge, we are the first to deliver a \emph{lossless} acceleration from a systems perspective.

\textbf{Disaggregated Serving.}
Disaggregated serving~\cite{distserve,mooncake,modserve} targets pipelines whose stages have contrasting resource profiles.
In LLM inference, \textit{prefill} is compute-bound while \textit{decode} is memory-bound; co-locating them degrades TTFT and TPOP, so systems such as \textsc{Mooncake} and \textsc{DistServe} separate these stages~\cite{mooncake,distserve}.
As multimodal inference matures, frameworks analogously isolate the \textit{encode} stage from \textit{prefill}/\textit{decode} to raise throughput, e.g., \textsc{EPD}~\cite{epd}.

\section{Conclusion and Future Works}
This work analyzes diffusion pipelines from both architectural and workload perspectives, revealing a pronounced imbalance in resource demands across requests and stages. To address this, we present \textsc{TridentServe}, the first system to support dynamic, stage-level resource allocation for both request processing and model deployment, jointly optimizing them to efficiently serve complex workloads.

As a future direction, we plan to adapt our system to \emph{heterogeneous} GPU pools, exploiting each stage’s distinct compute–memory profile to map it to the most cost-effective GPU tier. We also plan to extend this stage-level architecture to \emph{multi-model} serving.
Lastly, although this work is the first to articulate the problem of stage-level resource allocation, our current solution relies on certain heuristics (e.g., fixed model parallelism~\ref{subsec: extend model parallelism}), which we aim to refine in future work.

\bibliographystyle{ACM-Reference-Format}
\bibliography{software}
\newpage

\onecolumn
\appendix
\section{Impact of parallelism for other models.}
\label{sec: other impact effects}
Due to space limitations, we have only shown the impact of the parallelism strategy on the Flux.1 model in the main text. Here, we will demonstrate the impact of the parallelism strategy for \emph{Diffuse} stage on Stable-Diffusion-3, CogVideoX1.5-5B, and HunyuanVideo, As shown in Figure~\ref{fig:parallel_effect_extend}.

\section{Optimal Scheduling}
\label{sec:optimal scheduling}
In the main text (\S\ref{subsec: dispatch plan}) we stated that a theoretical optimal scheduling is a \emph{time-indexed, stage-level} scheduling problem, which is computationally intractable and thus impractical for online serving. 
We now give a complete formalization of it and followed by the linear programming equation. Then, we will prove that it is NP-completeness via a reduction to \emph{Job–Shop scheduling}, and analyze the model size to show that even modest instances are far beyond real-time solvability.

\subsection{Problem Formulation and Solution}
\label{app:optimal-formulation}
\underline{\textit{Problem data and notation.}}
Let \(\mathcal{G}=\{0,\dots,G-1\}\) be the GPU set, and let \(\pi_g\in\{\langle EDC\rangle,\langle DC\rangle,\langle ED\rangle,\langle D\rangle,\langle E\rangle,\langle C\rangle\}\) denote the \emph{placement} on GPU \(g\) (\S\ref{subsec: placement plan}). Let \(\mathcal{R}\) be the set of request batches; each \(r\in\mathcal{R}\) executes the fixed stage chain \(\{E\to D\to C\}\) and carries an SLO deadline \(d_r\). For every \(r\in\mathcal{R}\), \(s\in\mathcal{S}\), and parallel degree \(k\in    \{1,2,4,8\}\), a profiled processing time \(t_{r,s,k}>0\) is known (\S\ref{sec:engine-profmon}).  
Following \S\ref{sec:engine}, let \(\mathcal{W}\) be the catalog of pre-initialized \emph{intra-machine worker teams} (GPU combinations); each team \(w\in\mathcal{W}\) is a subset \(\mathrm{GPU}(w)\subseteq\mathcal{G}\) with \(|\mathrm{GPU}(w)|\) GPUs and is \emph{compatible} with stage \(s\) iff \(s\in\pi_g\) for all \(g\in\mathrm{GPU}(w)\). Write \(\mathcal{W}_s=\{w\in\mathcal{W}:\ s\ \text{compatible with}\ w\}\) and define the stage–team processing time \(t_{r,s}(w)\coloneqq t_{r,s,|\mathrm{GPU}(w)|}\).
Inter-stage communication times are known per request: \(Q_{r,ED}\ge 0\) for \(E{\to}D\) and \(Q_{r,DC}\ge 0\) for \(D{\to}C\) when the two stages do not run on the identical worker team; otherwise the comm.\ time is \(0\).

\underline{\textit{Decision variables.}}
For every \(r\in\mathcal{R}\), \(s\in\mathcal{S}\), and \(w\in\mathcal{W}_s\): assignment \(v_{r,s,w}\in\{0,1\}\) (choose exactly one team per stage); start/completion times \(S_{r,s}\ge 0\) and \(C_{r,s}\ge 0\). For co-location of consecutive stages, introduce \(z_{r,ED,w},z_{r,DC,w}\in\{0,1\}\) and define \(e_{r,ED}\coloneqq\sum_{w\in\mathcal{W}_E\cap\mathcal{W}_D}z_{r,ED,w}\) and \(e_{r,DC}\coloneqq\sum_{w\in\mathcal{W}_D\cap\mathcal{W}_C}z_{r,DC,w}\) (value \(1\) iff the two stages pick the same team). To prevent overlap on each GPU, use standard disjunctive ordering with binaries \(o_{(r,s),(r',s'),g}\in\{0,1\}\) for distinct operations and any \(g\in\mathcal{G}\). Finally \(y_r\in\{0,1\}\) indicates whether request \(r\) finishes before its deadline.

\underline{\textit{Objective.}}
\begin{align}
\max\ \Big(\ \sum_{r\in\mathcal{R}} y_r\ ,\ -\!\!\!\sum_{r\in\mathcal{R}}\big[Q_{r,ED}(1-e_{r,ED})+Q_{r,DC}(1-e_{r,DC})\big]\ \Big)
\tag{OBJ$'$}\label{eq:obj-prime}
\end{align}

\underline{\textit{Constraints.}}

\noindent\resizebox{\columnwidth}{!}{%
\begin{minipage}{\columnwidth}
\begin{align}
&\smash{\sum_{w\in\mathcal{W}_s}}\, v_{r,s,w}=1
&& \forall r\in\mathcal{R},\ \forall s\in\mathcal{S},
\tag{C$'$0a}\label{eq:assign-prime}\\[2pt]
&C_{r,s}=S_{r,s}+\smash{\sum_{w\in\mathcal{W}_s}} t_{r,s}(w)\,v_{r,s,w}
&& \forall r,s,
\tag{C$'$0b}\label{eq:dur-prime}\\[2pt]
&S_{r,D}\ \ge\ C_{r,E}\ +\ Q_{r,ED}\,(1-e_{r,ED})
&& \forall r,
\tag{C$'$1a}\label{eq:precED-prime}\\[2pt]
&S_{r,C}\ \ge\ C_{r,D}\ +\ Q_{r,DC}\,(1-e_{r,DC})
&& \forall r,
\tag{C$'$1b}\label{eq:precDC-prime}\\[2pt]
&z_{r,ED,w}\le v_{r,E,w}
&& \forall r,\ \ w\in\mathcal{W}_E\!\cap\!\mathcal{W}_D,
\tag{C$'$2a}\label{eq:colinED-ubE-prime}\\[-2pt]
&z_{r,ED,w}\le v_{r,D,w}
&& \forall r,\ \ w\in\mathcal{W}_E\!\cap\!\mathcal{W}_D,
\tag{C$'$2b}\label{eq:colinED-ubD-prime}\\[-2pt]
&z_{r,ED,w}\ge v_{r,E,w}+v_{r,D,w}-1
&& \forall r,\ \ w\in\mathcal{W}_E\!\cap\!\mathcal{W}_D,
\tag{C$'$2c}\label{eq:colinED-lb-prime}\\[2pt]
&z_{r,DC,w}\le v_{r,D,w}
&& \forall r,\ \ w\in\mathcal{W}_D\!\cap\!\mathcal{W}_C,
\tag{C$'$3a}\label{eq:colinDC-ubD-prime}\\[-2pt]
&z_{r,DC,w}\le v_{r,C,w}
&& \forall r,\ \ w\in\mathcal{W}_D\!\cap\!\mathcal{W}_C,
\tag{C$'$3b}\label{eq:colinDC-ubC-prime}\\[-2pt]
&z_{r,DC,w}\ge v_{r,D,w}+v_{r,C,w}-1
&& \forall r,\ \ w\in\mathcal{W}_D\!\cap\!\mathcal{W}_C,
\tag{C$'$3c}\label{eq:colinDC-lb-prime}\\[2pt]
&\begin{aligned}[t]
S_{r',s'} &\ \ge\ C_{r,s}\ -\ M\Big(3-o_{(r,s),(r',s'),g}\\[-2pt]
&\phantom{\ \ge\ C_{r,s}\ -\ M\Big(}
  - \smash{\sum_{w\ni g}} v_{r,s,w}\ - \smash{\sum_{w'\ni g}} v_{r',s',w'}\Big)
\end{aligned}
&& \forall g\in\mathcal{G},\ \forall (r,s)\neq(r',s'),
\tag{C$'$4a}\label{eq:disj-a-prime}\\[2pt]
&\begin{aligned}[t]
S_{r,s} &\ \ge\ C_{r',s'}\ -\ M\Big(2- o_{(r,s),(r',s'),g}\\[-2pt]
&\phantom{\ \ge\ C_{r',s'}\ -\ M\Big(}
  - \smash{\sum_{w'\ni g}} v_{r',s',w'}\Big)
\end{aligned}
&& \forall g\in\mathcal{G},\ \forall (r,s)\neq(r',s'),
\tag{C$'$4b}\label{eq:disj-b-prime}\\[2pt]
&C_{r,C}\ \le\ d_r + M(1-y_r)
&& \forall r,
\tag{C$'$5}\label{eq:ddl-prime}\\[2pt]
&v_{r,s,w},\ z_{r,ED,w},\ z_{r,DC,w},\ o_{(r,s),(r',s'),g},\ y_r\in\{0,1\}.
&&
\tag{C$'$6}\label{eq:dom-prime}
\end{align}
\end{minipage}%
}

\noindent\textbf{(OBJ$'$)} Lexicographically maximizes the number of on-time completions \(\sum_{r\in\mathcal{R}} y_r\) and, among maximizers, minimizes inter-stage communication cost via \(Q_{r,ED}(1-e_{r,ED})+Q_{r,DC}(1-e_{r,DC})\).

\noindent\textbf{(C$'$0a)} Assigns exactly one compatible worker team to each stage of each request.

\noindent\textbf{(C$'$0b)} Defines each stage’s completion time as its start time plus the profiled duration \(t_{r,s}(w)\) of the chosen team.

\noindent\textbf{(C$'$1a–C$'$1b)} Enforce the stage precedence \(E\!\to\!D\!\to\!C\) and add the corresponding inter-stage delay \(Q_{r,\cdot\cdot}\) whenever the two consecutive stages are not co-located (i.e., \(e_{r,\cdot\cdot}\neq 1\)).

\noindent\textbf{(C$'$2a–C$'$2c)} Linearize co-location for \(E\) and \(D\) so that \(z_{r,ED,w}=1\) if and only if both stages select the same team \(w\), which yields \(e_{r,ED}=\sum_{w} z_{r,ED,w}\).

\noindent\textbf{(C$'$3a–C$'$3c)} Apply the analogous co-location linearization to \(D\) and \(C\), giving \(e_{r,DC}=\sum_{w} z_{r,DC,w}\).

\noindent\textbf{(C$'$4a–C$'$4b)} Impose big-\(M\) disjunctive no-overlap on every GPU by ordering any pair of operations that may require that GPU.

\noindent\textbf{(C$'$5)} Link the on-time indicator \(y_r\) to the deadline \(d_r\) by permitting \(y_r=1\) only if the request finishes no later than \(d_r\).

\noindent\textbf{(C$'$6)} Declare all assignment, co-location, ordering, and deadline indicators to be binary variables.

\noindent The constant \(M\) upper-bounds the planning horizon; for example, one may take \(M \ge \sum_{r,s}\max_{w\in\mathcal{W}_s} t_{r,s}(w) + \max_{r} d_r\).

\subsection{NPC Proof}\label{app:npc-proof}

We prove that the decision version of the stage-aware scheduling model defined in Appendix~\ref{app:optimal-formulation} is NP-complete.

\begin{figure*}[!t]
    \centering
    \includegraphics[width=\linewidth]{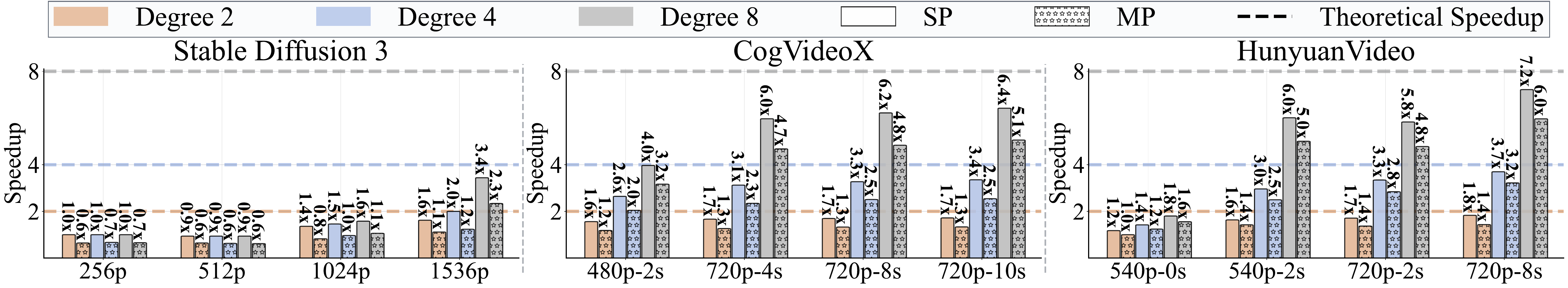}
    \caption{\small{Impact of parallelism for other models.
    }}
    \label{fig:parallel_effect_extend}
\end{figure*}

\medskip
\noindent\emph{\textbf{SADP-Deadline (decision).}}  Given the data and notation in Appendix~\ref{app:optimal-formulation}, we consider the following decision problem:
Given a common deadline $T>0$ (i.e., set $d_r:=T$ for all $r\in\mathcal{R}$), decide whether there exist assignments of the variables that satisfy the feasibility constraints (assignment, duration, precedence with inter-stage time, no-overlap per GPU, integrality—cf.\ constraints \eqref{eq:assign-prime}–\eqref{eq:dom-prime}) and achieve $C_{r,C}\le T$ for all $r\in\mathcal{R}$ (equivalently, $y_r=1$ for all $r$ under \eqref{eq:ddl-prime}).

\begin{proposition}[NP-completeness]\label{prop:npc-restate}
\textsc{SADP-Deadline} is NP-complete. The hardness holds even under the restricted setting in which (i) $\mathcal{W}_E,\mathcal{W}_D,\mathcal{W}_C$ each contain only single-GPU teams, (ii) $Q_{r,ED}=Q_{r,DC}=0$ for all $r$, and (iii) placements $\pi_g$ forbid co-location of different stages on the same GPU (i.e., $\pi_g\in\{\langle E\rangle,\langle D\rangle,\langle C\rangle\}$).
\end{proposition}

\begin{proof}
\textbf{Membership in NP.}
Given a candidate solution \{$v_{r,s,w}$, $S_{r,s}$,$C_{r,s}$,$z_{r,ED,w}$,$z_{r,DC,w}$,$o_{(r,s),(r',s'),g}$,$y_r$\}, verification of \eqref{eq:assign-prime}–\eqref{eq:dom-prime} is polynomial in the input size: one checks per-stage assignment (\eqref{eq:assign-prime}), durations (\eqref{eq:dur-prime}), precedence with inter-stage time (\eqref{eq:precED-prime}–\eqref{eq:precDC-prime}), co-location linearization (\eqref{eq:colinED-ubE-prime}–\eqref{eq:colinDC-lb-prime}), no-overlap on each GPU via the disjunctive constraints (\eqref{eq:disj-a-prime}–\eqref{eq:disj-b-prime}), and deadlines (\eqref{eq:ddl-prime}). Hence the problem is in NP.

\medskip
\textbf{NP-hardness via reduction from three-machine flow shop.}
Reduce from the following decision problem: given $n$ jobs that must be processed non-preemptively on three machines in the fixed order $M_1\!\to M_2\!\to M_3$ with processing times $(a_{j,1},a_{j,2},a_{j,3})$ for job $j\in\{1,\dots,n\}$ and a threshold $T>0$, decide whether there exists a schedule with makespan at most $T$. This problem is a special case of job shop scheduling and is NP-complete.

\medskip
\noindent\emph{Construction of the \textsc{SADP-Deadline} instance.}
From the flow-shop instance, build an instance using the objects already defined in Appendix~\ref{app:optimal-formulation} as follows:

\begin{itemize}
  \item \textbf{Placements and teams.}  
  Create exactly three GPUs: $g_E,g_D,g_C$ with $\pi_{g_E}=\langle E\rangle$, $\pi_{g_D}=\langle D\rangle$, $\pi_{g_C}=\langle C\rangle$.  
  Define $\mathcal{W}_E=\{\{g_E\}\}$, $\mathcal{W}_D=\{\{g_D\}\}$, $\mathcal{W}_C=\{\{g_C\}\}$ and $\mathcal{W}_s=\emptyset$ for any other $s$ (none).  
  Thus each stage type has unit capacity and only single-GPU teams are available.
  \item \textbf{Requests and processing times.}  
  Set $\mathcal{R}=\{r_1,\dots,r_n\}$.  
  For each $r_j$ and the unique compatible team $w_E\in\mathcal{W}_E$, $w_D\in\mathcal{W}_D$, $w_C\in\mathcal{W}_C$, set
  \[
  t_{r_j,E}(w_E):=a_{j,1},\qquad
  t_{r_j,D}(w_D):=a_{j,2},\qquad
  t_{r_j,C}(w_C):=a_{j,3}.
  \]
  (Since $t_{r,s}(w)=t_{r,s,|\mathrm{GPU}(w)|}$ by definition, this pins the times for the only admissible $w$ per stage.)
  \item \textbf{Inter-stage times and deadlines.}  
  Set $Q_{r,ED}=Q_{r,DC}=0$ for all $r\in\mathcal{R}$ and let $d_r:=T$ for all $r$ (common deadline).
\end{itemize}

\noindent\emph{Correctness of the reduction.}
Under this construction, any feasible solution to \eqref{eq:assign-prime}–\eqref{eq:dom-prime} selects for each $r$ exactly one team per stage (the unique available one), enforces the fixed precedence $E\!\to D\!\to C$, and prevents overlap on each of the three single-capacity stage resources via \eqref{eq:disj-a-prime}–\eqref{eq:disj-b-prime}.  
Because $Q_{r,ED}=Q_{r,DC}=0$ and stages cannot co-locate (by placement), each request’s completion time is exactly the finish time of its $C$ stage, with per-stage durations matching $(a_{j,1},a_{j,2},a_{j,3})$.  
Therefore, the machine sequence $(M_1,M_2,M_3)$ in the flow shop is faithfully represented by the stage resources $(E,D,C)$ here, each with unit capacity and identical processing times.

Consequently, there exists a flow-shop schedule with makespan $\le T$ if and only if there exist variables satisfying \eqref{eq:assign-prime}–\eqref{eq:dom-prime} with $C_{r,C}\le T$ for all $r\in\mathcal{R}$ in the constructed instance (equivalently, $y_r=1$ for all $r$ by \eqref{eq:ddl-prime}).  
The transformation is clearly polynomial in the input size.

\medskip
\textbf{Conclusion.}
\textsc{SADP-Deadline} is NP-hard by the reduction above and in NP by the verification argument; hence it is NP-complete.  
Moreover, hardness holds already under the restricted setting specified in the statement, and therefore a fortiori for the full model with general $\mathcal{W}$, nonzero $Q_{\cdot,\cdot}$, and richer placements $\pi_g\in\{\langle EDC\rangle,\langle DC\rangle,\langle ED\rangle,\langle D\rangle,\langle E\rangle,\langle C\rangle\}$.
\end{proof}

\paragraph{Relation to job shop scheduling.}
The three-machine flow shop is a special case of job shop scheduling in which all jobs share the same route; the reduction above embeds that special case into our model using the already-defined objects $(\mathcal{G},\pi,\mathcal{R},\mathcal{W},t_{r,s}(\cdot),Q_{\cdot,\cdot})$, thereby establishing NP-hardness relative to job shop as well.

\subsection{Complexity Analysis}
\label{app:complexity}
We size the exact optimal ILP model in Appendix~\ref{app:optimal-formulation} and show that, even at very small scales, it is already unacceptable in the online service scenario.

\paragraph{Where the blow-up comes from.}
Let \(R\coloneqq|\mathcal{R}|\) be the number of (pending) request batches at a tick, \(S\coloneqq|\mathcal{S}|=3\) the number of stages (\(E,D,C\)), and \(G\coloneqq|\mathcal{G}|\) the number of GPUs. The formulation \eqref{eq:assign-prime}–\eqref{eq:dom-prime} includes, for every \emph{ordered} pair of distinct stage-operations and for every GPU \(g\in\mathcal{G}\), the disjunctive no-overlap binaries and constraints \eqref{eq:disj-a-prime}–\eqref{eq:disj-b-prime}. The number of stage-operations is \(RS\); the number of unordered pairs is \(\binom{RS}{2}\). Hence the disjunctive layer alone contributes
\[
\underbrace{G\,\binom{RS}{2}}_{\text{binaries }o_{(\cdot),(\cdot),g}}
\qquad\text{and}\qquad
\underbrace{2G\,\binom{RS}{2}}_{\text{linear constraints}},
\]
which scales as \(\Theta(G\,R^{2}S^{2})\) and dominates the model size. In comparison, the remaining blocks—assignment variables \(v_{r,s,w}\), co-location linearizations \(z_{r,\cdot,w}\) (cf.\ \eqref{eq:colinED-ubE-prime}–\eqref{eq:colinDC-lb-prime}), and timing/deadline constraints \eqref{eq:dur-prime}, \eqref{eq:precED-prime}–\eqref{eq:precDC-prime}, \eqref{eq:ddl-prime}—grow only linearly in \(R\) (and in the team catalogs \(|\mathcal{W}_s|\)) and are negligible next to the quadratic disjunctive term.

\paragraph{A small numeric instance.}
Take a modest online tick with \(R=20\) pending batches on a common evaluation cluster of size \(G=128\), and \(S=3\).
Then \(RS=60\) and
\[
\binom{RS}{2}=\binom{60}{2}=\frac{60\cdot 59}{2}=1770,
\\
G\,\binom{RS}{2}=128\times 1770=226{,}560.
\]
Thus, the disjunctive layer already yields \(226{,}560\) binary variables \(o_{(\cdot),(\cdot),g}\) and \(2\times 226{,}560=453{,}120\) linear constraints—\emph{before} counting any assignment/co-location binaries \(v,z\), start/finish times \(S_{r,s},C_{r,s}\), or deadlines \(y_r\).
Even if \(|\mathcal{W}_s|\) is small, the root LP at this scale contains several hundred thousand rows/columns; in practice, presolve and the first LP relaxation typically take on the order of \(\mathbf{0.5\text{–}3\,s}\) on commodity servers, while the full MILP (with branch-and-bound and cuts) often ranges from \(\mathbf{several\ seconds}\) to \(\mathbf{minutes}\) depending on \(R\), \(G\).

\paragraph{Consequence for online scheduling.}
Besides NPC and the sheer model size, there are two practical blockers for online serving:
\begin{enumerate}
  \item \textbf{Solve time.} Building and solving the exact MILP per tick cannot meet a \(\mathcal{O}(10^2)\) ms budget: even for \(R\!\approx\!20\), the root relaxation already incurs \(\gtrsim\)sub-second latency, and end-to-end MILP solve time is typically seconds to minutes.
  \item \textbf{Brittleness of time-index decisions.} The optimal model commits to \emph{concrete start times} \(S_{r,s}\) for all stages. In online serving, runtime jitter from GPU contention, network variation, kernel variability, and queueing introduces non-negligible timing noise; long-horizon, time-coupled decisions thus become quickly inaccurate. This causes error accumulation (mis-ordered plans and infeasible overlaps) and poor realized performance when the system state deviates from the assumed schedule.
\end{enumerate}
These observations motivate our design in the main text: (i) \emph{decouple} placement and dispatch, (ii) adopt a \emph{myopic, Diffuse-first} ILP with aggressive feasibility filtering to keep per-tick solve time in the sub-second regime.

\section{Implementation detail of Planner}
\subsection{Placement Plan}
\label{subsec: placement plan detail}
\textbf{\small \texttt{Split()}}. For each Virtual-Replica type $t\!\in\!\{\mathrm{EDC},\mathrm{ED},\mathrm{DC},\mathrm{D}\}$ with GPU budget $N_t$ and per\mbox{-}replica service rates of the relevant \emph{Primary} and \emph{Auxiliary} roles under this type ($v^{\mathrm{prim}}_t$, and when applicable $v^{\mathrm{auxE}}_t$, $v^{\mathrm{auxC}}_t$), we produce integer counts $(n^{\mathrm{prim}}_t,n^{\mathrm{auxE}}_t,n^{\mathrm{auxC}}_t)$ summing to $N_t$ so that auxiliary capacity is sufficient for what the primary processes. 
\begin{itemize}
    \item For $\mathrm{EDC}$ the split is trivial: all $N_t$ go to the primary ($n^{\mathrm{prim}}_t\!=\!N_t$). 
    \item For $\mathrm{ED}$ we match $\langle\mathrm{C}\rangle$ to the primary by inverse\mbox{-}proportion: with $\rho\!=\!v^{\mathrm{prim}}_t/v^{\mathrm{auxC}}_t$, set $n^{\mathrm{prim}}_t\!=\!\lfloor N_t/(1+\rho)\rfloor$ and $n^{\mathrm{auxC}}_t\!=\!N_t-n^{\mathrm{prim}}_t$ (and $n^{\mathrm{auxE}}_t\!=\!0$), which guarantees $n^{\mathrm{auxC}}_t v^{\mathrm{auxC}}_t\!\ge\! n^{\mathrm{prim}}_t v^{\mathrm{prim}}_t$.
    \item For $\mathrm{DC}$ we do the symmetric construction with $\rho\!=\!v^{\mathrm{prim}}_t/v^{\mathrm{auxE}}_t$.
    \item For $\mathrm{D}$ we need both auxiliaries; we allocate proportionally to the rate ratios $a\!=\!v^{\mathrm{prim}}_t/v^{\mathrm{auxE}}_t$ and $b\!=\!v^{\mathrm{prim}}_t/v^{\mathrm{auxC}}_t$ by forming the real vector $\tfrac{N_t}{1+a+b}(1,a,b)$ for $(\mathrm{prim},\mathrm{auxE},\mathrm{auxC})$ and rounding to integers. 
\end{itemize}
   If rounding creates a budget overflow or violates the obvious lower bounds ($n^{\mathrm{auxE}}_t v^{\mathrm{auxE}}_t\!\ge\! n^{\mathrm{prim}}_t v^{\mathrm{prim}}_t$ and $n^{\mathrm{auxC}}_t v^{\mathrm{auxC}}_t\!\ge\! n^{\mathrm{prim}}_t v^{\mathrm{prim}}_t$), we decrease $n^{\mathrm{prim}}_t$ by one and reassign the remaining unit to the auxiliary with the largest deficit; tiny budgets fall back to assigning the available units to the roles present, prioritising feasibility over exact proportionality.

\textbf{\small \texttt{PackPerMachine()}}. After all types are split, we first align the number of \emph{Primary} replicas \emph{that contain $\mathrm{D}$} to multiples of $8$ so as not to cap sequence\mbox{-}parallel choices ($k\!\le\!8$): for each $t\!\in\!\{\mathrm{EDC},\mathrm{ED},\mathrm{DC},\mathrm{D}\}$ we increase $n^{\mathrm{prim}}_t$ up to the next multiple of $8$ by borrowing from its auxiliaries while keeping the \texttt{Split} feasibility bounds; if the borrow is infeasible we leave $n^{\mathrm{prim}}_t$ as is. 

We then place replicas onto 8\mbox{-}GPU nodes by packing homogeneous blocks: primaries of the same $\pi_g\!\in\!\{\langle \mathrm{EDC}\rangle,\langle \mathrm{DC}\rangle,\langle \mathrm{ED}\rangle,\langle \mathrm{D}\rangle\}$ fill whole nodes whenever possible, auxiliaries $\langle\mathrm{E}\rangle$ and $\langle\mathrm{C}\rangle$ are packed next in the same manner, and any remainders are assigned by a simple first\mbox{-}fit that prefers nodes already hosting the same $\pi_g$. This produces $\mathcal{P}$ with D\mbox{-}carrying primaries padded to $8$ and high intra\mbox{-}node homogeneity.

\subsection{Dispatch Plan}
\label{subsec: dispatch plan details}
In solving the ILP for the dispatch plan, our goal is to maximize the SLO satisfaction rate while try to minimize inter-stage communication. 
To achieve this, we set a reward weight for task completion $W$ and a communication penalty weight $Q$. Below, we will detail the settings for these two parameters.

We instantiate the objective weights to strongly prioritize SLO attainment and only then prefer lower inter\mbox{-}stage communication. Let
\begin{equation}
  \widehat{T}_r \;\triangleq\; \min_{\substack{i\in\mathcal{I},\,k\in\mathcal{K}:\\ E_{r,k}F_{r,i}=1}}\bigl(\tau + t_{r,i,k}\bigr), 
  \qquad
  \mathrm{scale}_r \;\triangleq\; \max\!\Bigl\{1,\;\frac{\widehat{T}_r}{d_r}\Bigr\},
\end{equation}
be, respectively, the best predicted completion time if $r$ were dispatched now and its overtime factor. The completion reward is
\begin{equation}
  W_r \;=\;
  \begin{cases}
    C_{\mathrm{on}}, & \widehat{T}_r \le d_r \quad\text{(on time)},\\[4pt]
    C_{\mathrm{late}} \cdot \max\!\bigl\{1,\, \mathrm{scale}_r - \alpha + 1\bigr\}, 
      & \widehat{T}_r > d_r \quad\text{(late)},
  \end{cases}
  \label{eq:Wr}
\end{equation}
where $\alpha$ is the starvation threshold. This realizes the intended behavior: for example, with $\alpha{=}5$, if $\mathrm{scale}_r{=}5$ the reward remains $C_{\mathrm{late}}$; if $\mathrm{scale}_r{=}6$ it becomes $2C_{\mathrm{late}}$; if $\mathrm{scale}_r{=}2$ it stays $C_{\mathrm{late}}$.

The communication penalty depends on the Primary type $i$ and scales linearly with the batch’s processing length $l_r$:
\begin{equation}
  Q_{r,i} \;=\;
  \begin{cases}
    \beta_{0}\, l_r, & i=0\ \ (\langle EDC\rangle),\\
    \beta_{1}\, l_r, & i=1\ \ (\langle DC\rangle),\\
    \beta_{2}\, l_r, & i=2\ \ (\langle ED\rangle),\\
    \beta_{3}\, l_r, & i=3\ \ (\langle D\rangle).
  \end{cases}
  \label{eq:Qr}
\end{equation}
By construction, the communication penalty is intentionally orders of magnitude smaller than the completion reward so that the solver first maximizes the SLO satisfaction rate; among choices with comparable on\mbox{-}time prospects, it then prefers placements that incur less inter\mbox{-}stage transfer.

\paragraph{Constant values.}
Unless otherwise stated, we use $C_{\mathrm{on}}{=}1000$, $C_{\mathrm{late}}{=}200$, $\alpha{=}5$, and $(\beta_{0},\beta_{1},\beta_{2},\beta_{3}){=}\bigl(0,\,10^{-6},\,5{\times}10^{-6},\,6{\times}10^{-6}\bigr)$.

\section{Details in Experiments Settings}
\subsection{Workload Setting}
\label{subsec: workload setting}
We also use the number of denoising steps per model to the officially recommended~\cite{blackforestlabs_flux1_schnell_hf,fastvideo_fast_hunyuan_hf,cogvideoxmodel,stablediffusion3} or commonly used value~\cite{dmd} in real scene. Because different models have different compute demands, we also set model-specific request rates so that the workloads fit our \(128\)-GPU cluster. The full configuration is summarized in Table~\ref{tab:workload-mix}.

\textbf{Steady Workload.}
Based on each model’s recommended output resolutions and (for video) duration limits, we construct three \emph{steady} workload \emph{Light}, \emph{Medium}, and \emph{Heavy}—tailored to every model. Details is also shown in Table~\ref{tab:workload-mix}.

\textbf{Proprietary Workload Scaling.}
In addition to the synthetic workload above, we use two \emph{proprietary} traces (one for image generation and one for video generation), whose patterns have been shown in Figure~\ref{fig:workload pattern}. Applying these traces directly in our experiments would either overprovision or overload our evaluation setup due to differences in real case cluster/model and ours.
Therefore, we scale each \emph{proprietary} trace to match the total number of requests in a \(30\,\mathrm{min}\) window to that of the corresponding \emph{Steady} workload while preserving the temporal pattern. Concretely, if the \emph{proprietary} trace contains too many requests, we uniformly subsample according to its native distribution; if it contains too few, we replicate requests proportionally and round up to integers.
\newcolumntype{Y}{>{\RaggedRight\arraybackslash}X}

\begin{table*}[t]
\setlength{\tabcolsep}{5pt}
\renewcommand{\arraystretch}{1.12}
\small
\centering
\caption{Resolutions, video length mixed ratios in different \emph{steady} workload types, and request rates, denoising steps, $T_\text{win}$ for different models. ``$k\times\{\cdot\}$'' are relative weights (normalized when sampling).}
\label{tab:workload-mix}
\begin{tabularx}{\textwidth}{@{}l c c c l Y@{}}
\toprule
Model & Rate (req/s) & Steps & \(T_{\text{win}}\) & Type & Resolution/Duration mix ratio (compact weights) \\
\midrule
\multirow{3}{*}{StableDiffusion3-Medium} & \multirow{3}{*}{20} & \multirow{3}{*}{20} & \multirow{3}{*}{$3\,\mathrm{min}$}
& Light  & \makecell[l]{$2\times\{128^2,\,256^2\}$; \\ $1\times\{512^2,\,1024^2,\,1536^2\}$} \\
& & & & Medium & \makecell[l]{$4\times\{512^2\}$; \\ $1\times\{128^2,\,256^2,\,1024^2,\,1536^2\}$} \\
& & & & Heavy  & \makecell[l]{$2\times\{1024^2,\,1536^2\}$; \\ $1\times\{128^2,\,256^2,\,512^2\}$} \\
\addlinespace[2pt]
\midrule
\multirow{3}{*}{Flux.1} & \multirow{3}{*}{1.5} & \multirow{3}{*}{4} & \multirow{3}{*}{$5\,\mathrm{min}$}
& Light  & \makecell[l]{$2\times\{128^2,\,256^2,\,512^2\}$; \\ $1\times\{1024^2,\,2048^2,\,3072^2,\,4096^2\}$} \\
& & & & Medium & \makecell[l]{$2\times\{1024^2,\,2048^2\}$; \\ $1\times\{128^2,\,256^2,\,512^2,\,3072^2,\,4096^2\}$} \\
& & & & Heavy  & \makecell[l]{$2\times\{3072^2,\,4096^2\}$; \\ $1\times\{128^2,\,256^2,\,512^2,\,1024^2,\,2048^2\}$} \\
\addlinespace[2pt]
\midrule
\multirow{3}{*}{CogVideoX1.5-5B} & \multirow{3}{*}{1} & \multirow{3}{*}{6} & \multirow{3}{*}{$5\,\mathrm{min}$}
& Light  & \makecell[l]{$3\times\{480\text{p}--2\text{s},\,720\text{p}--2\text{s}\}$; \\ $1\times\{480\text{p}--\{4,8,10\}\text{s},\,720\text{p}--\{4,8,10\}\text{s}\}$} \\
& & & & Medium & \makecell[l]{$2\times\{480\text{p}--\{4,8,10\}\text{s}\}$; \\ $1\times\{480\text{p}--2\text{s},\,720\text{p}--2\text{s},\,720\text{p}--\{4,8,10\}\text{s}\}$} \\
& & & & Heavy  & \makecell[l]{$2\times\{720\text{p}--\{4,8,10\}\text{s}\}$; \\ $1\times\{480\text{p}--2\text{s},\,720\text{p}--2\text{s},\,480\text{p}--\{4,8,10\}\text{s}\}$} \\
\addlinespace[2pt]
\midrule
\multirow{3}{*}{HunyuanVideo} & \multirow{3}{*}{0.5} & \multirow{3}{*}{6} & \multirow{3}{*}{$10\,\mathrm{min}$}
& Light  & \makecell[l]{$3\times\{540\text{p}--1\text{s},\,720\text{p}--1\text{s}\}$; \\ $1\times\{540\text{p}--\{2,4,8\}\text{s},\,720\text{p}--\{2,4,8\}\text{s}\}$} \\
& & & & Medium & \makecell[l]{$2\times\{540\text{p}--\{2,4\}\text{s},\,720\text{p}--2\text{s}\}$; \\ $1\times\{540\text{p}--1\text{s},\,720\text{p}--1\text{s},\,720\text{p}--4\text{s},\,540/720\text{p}--8\text{s}\}$} \\
& & & & Heavy  & \makecell[l]{$2\times\{720\text{p}--4\text{s},\,540/720\text{p}--8\text{s}\}$; \\ $1\times\{540\text{p}--1\text{s},\,720\text{p}--1\text{s},\,540\text{p}--\{2,4\}\text{s},\,720\text{p}--2\text{s}\}$} \\
\bottomrule
\end{tabularx}
\end{table*}

\subsection{Baselines Implementation}
\label{subsec: baseline implementation}
\textbf{B1 (Static Pipeline-level).} We set a single, global sequence-parallel degree $k$ per model to half of the model’s optimal degree at its \emph{maximum} load length. This choice is consistent with our SLO definition (SLO $=2.5\times$ the latency at the optimal degree): using $k_{\max}/2$ ensures even longest requests still meet the SLO while avoiding the resource waste of always running at the highest feasible degree. Concretely, we use $k{=}2$ for \textit{Sd3} and $k{=}4$ for \textit{Flux}, \textit{Cog}, and \textit{Hunyuan}. The concrete parallelism strategy we choose is Ulysses-k.

\textbf{B2/B5 (Bucketed).} We pre-partition the cluster into degree buckets (e.g., $k\in\{1,2,4,8\}$ as supported by each model) and route each request to the bucket matching its profiled-optimal degree. Let $r_k$ be the target share for degree $k$ (per model and workload level) and let $N{=}128$ be the total available GPUs; we allocate \emph{GPU counts} $N_k$ by $N_k=\mathrm{round\_to\_mult}(N\cdot r_k,\;k)$, where $\mathrm{round\_to\_mult}(x,k)$ rounds $x$ to the nearest multiple of $k$ (ties downward). Finally we adjust $N_1\!\leftarrow\!N-\sum_{k\in\{2,4,8\}}N_k$ so that $\sum_k N_k{=}128$. We show the example for \emph{Steady} workload \emph{GPU} allocations in Table~\ref{tab:bucket-shares}. 

\textbf{B4/B6 (SRTF with aging).} We schedule ready requests by Shortest-Remaining-Time-First (SRTF) using profiler estimates, while aging overdue requests to avoid starvation. For request $r$ with deadline $d_r$ and estimated completion $\widehat T_r$, if $\widehat T_r\le d_r$ it goes to the top-priority queue; otherwise we compute an integer overtime scale $\mathrm{scale}_r=\bigl\lceil(\widehat T_r-d_r)/t_r^*\bigr\rceil$ where $t_r^*$ is the profiled runtime at the optimal degree, and assign priority $p_r=\max\!\bigl(1,\,5-\mathrm{scale}_r\bigr)$ (smaller $p_r$ is higher priority). Within each priority we still pick by SRTF and refresh $\widehat T_r$ when $k$ or placement changes. B4 applies this to colocated instances; B6 applies it within disaggregated clusters.

\textbf{B5/B6 (Stage-cluster sizing).} For disaggregated serving, we split the $G{=}128$ GPUs in inverse proportion to their measured per-instance service rates. With stage throughputs $\hat v_E,\hat v_D,\hat v_C$ and optional demand weights $w_s$ (default $1$), we set
\[
p_s \;=\; \frac{w_s/\hat v_s}{\sum_{s'\in\{E,D,C\}} w_{s'}/\hat v_{s'}}\,,\qquad
G_s \;=\; \mathrm{round}(G\cdot p_s)\,,
\]
then adjust the largest $G_s$ by $\pm1$ if needed to ensure $\sum_s G_s{=}128$. The concrete \emph{GPU} splits used are in Table~\ref{tab:stage-splits}.

\begin{table*}[t]
\centering
\caption{Bucketed GPU allocations $N_k$ for B2/B5 with total $N{=}128$ GPUs per model \& workload. Each $N_k$ is padded to a multiple of its degree $k$; the $k{=}1$ bucket is adjusted last so that rows sum to $128$.}
\label{tab:bucket-shares}
\setlength{\tabcolsep}{6pt}
\renewcommand{\arraystretch}{1.12}
\begin{tabular}{@{}l l *{4}{c} c@{}}
\toprule
\multirow{2}{*}{Model} & \multirow{2}{*}{Workload} & \multicolumn{4}{c}{\textbf{GPUs @ degree $k$}} & \multirow{2}{*}{Total} \\
\cmidrule(lr){3-6}
 &  & $k{=}8$ & $k{=}4$ & $k{=}2$ & $k{=}1$ &  \\
\midrule
\multirow{3}{*}{SD3-M} 
 & Light  & 0  & 16 & 18 & 94 & 128 \\
 & Medium & 0  & 16 & 16 & 96 & 128 \\
 & Heavy  & 0  & 36 & 36 & 56 & 128 \\
\addlinespace[3pt]
\multirow{3}{*}{CogVideoX1.5-5B} 
 & Light  & 8  & 12 & 64 & 44 & 128 \\
 & Medium & 8  & 12 & 70 & 38 & 128 \\
 & Heavy  & 24 & 24 & 58 & 22 & 128 \\
\addlinespace[3pt]
\multirow{3}{*}{Flux.1}
 & Light  & 16 & 12 & 12 & 88 & 128 \\
 & Medium & 16 & 16 & 28 & 68 & 128 \\
 & Heavy  & 32 & 28 & 14 & 54 & 128 \\
\addlinespace[3pt]
\multirow{3}{*}{HunyuanVideo}
 & Light  & 40 & 12 & 44 & 32 & 128 \\
 & Medium & 56 & 24 & 36 & 12 & 128 \\
 & Heavy  & 72 & 12 & 32 & 12 & 128 \\
\bottomrule
\end{tabular}
\end{table*}

\begin{table*}[t]
\centering
\caption{Stage-level GPU splits $G_s$ for B5/B6 with $G{=}128$ GPUs. Values are rounded to integers and adjusted to sum to $128$.}
\label{tab:stage-splits}
\setlength{\tabcolsep}{10pt}
\renewcommand{\arraystretch}{1.12}
\begin{tabular}{@{}l ccc c@{}}
\toprule
\multirow{2}{*}{Model} & \multicolumn{3}{c}{\textbf{GPUs per stage}} & \multirow{2}{*}{Total} \\
\cmidrule(lr){2-4}
 & Encode $(G_E)$ & Diffuse $(G_D)$ & Decode $(G_C)$ & \\
\midrule
SD3-M         & 12 & 88  & 28 & 128 \\
Flux.1        & 6 & 96  & 26 & 128 \\
CogVideoX1.5-5B  & 4  & 102 & 18 & 128 \\
HunyuanVideo  & 4  & 112 & 12 & 128 \\
\bottomrule
\end{tabular}
\end{table*}

\section{Extend for our method in using Batch Size and Model Parallel.}
\label{sec: extend method}
\subsection{Batch Size}
\label{subsec: extend batch size}
Our method can be directly integrated with dynamic batch without any changes. We first analyze the impact of different batch sizes on requests and stages, and then briefly explain how our method can be integrated with dynamic batch size with minimal changes.

\textbf{Impact of Batch Size.}
As shown in Figure~\ref{fig:batch effect}, the \emph{encoder} handles lightweight inputs, so it can be run with very large batches almost without slowing down.
The \emph{diffusion model} is compute-bound, batching helps only when low-resolution images are generated.
The \emph{decoder} are memory-bound, and their latency grows almost linearly with batch size, so batching provides no benefit.

\begin{figure}[H]
    \centering
    \includegraphics[width=0.7\linewidth]{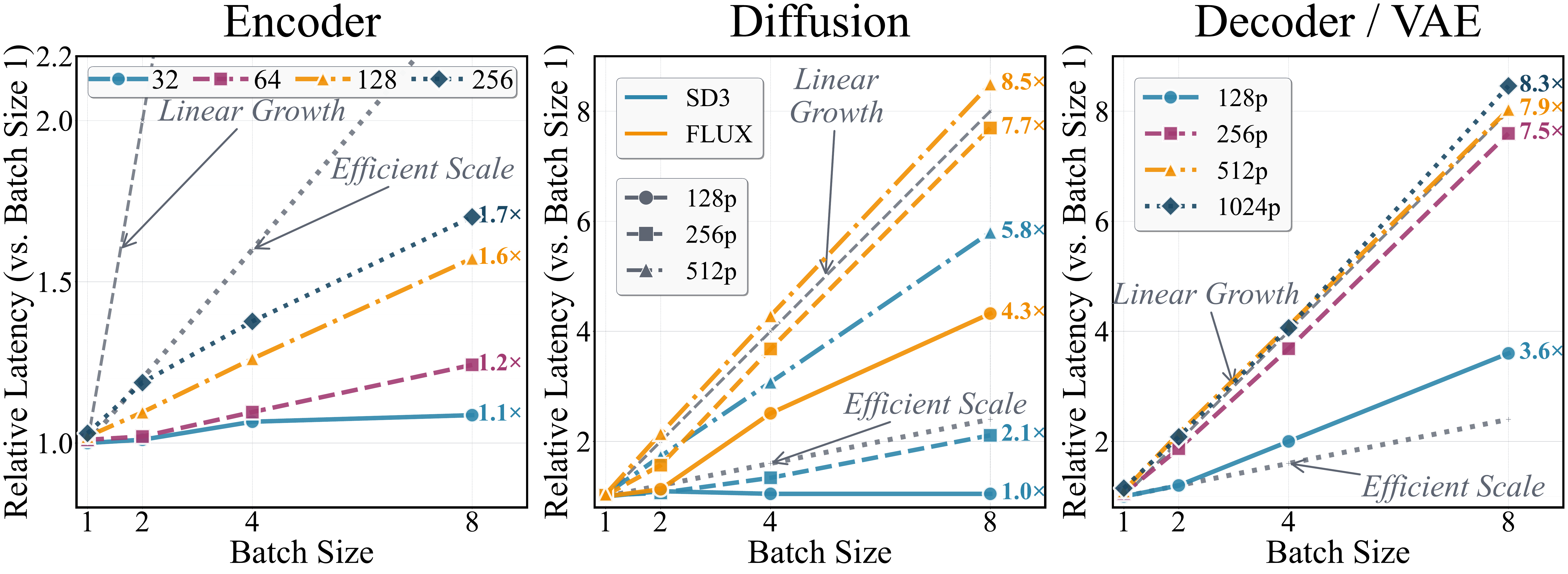}
    \caption{Batch effects on different stages. Encoder: \emph{T5-XXL}; Decoder: \emph{AutoEncoderKL}. Batching is efficient when latency overhead $\le 20\%$.}
    \label{fig:batch effect}
\end{figure}
\textbf{Integration with Dynamic Batch.} From the analysis above, we conclude that different requests and different stages each have distinct optimal batch sizes (we define the optimal batch size as the largest batch size whose latency increase is less than 20\%). We observe that batch scalability follows \emph{Encode} $>$ \emph{Diffuse} $>$ \emph{Decode}. Since the \emph{Diffuse} stage dominates runtime, we use its optimum as the standard: for requests of the same size, we form batches using the \emph{Diffuse} stage’s optimal batch size, and then perform resource allocation and execution at the granularity of the request batch; the method requires virtually no changes.

Meanwhile, because the \emph{Encode} stage’s optimal batch size exceeds that of the \emph{Diffuse} stage, some $\Gamma^E$ may not reach their optimal batch size. We proactively merge the $\Gamma^E$ that run exclusively on $\pi = \langle E \rangle$, consolidating those that have not yet reached the optimal batch size, thereby further improving resource utilization.

\subsection{Model Parallelism}
\label{subsec: extend model parallelism}
Because model parallelism is far less efficient than sequence parallelism at the same degree, we enable model parallelism only when the model parameters cannot fit on a single GPU. We enable model parallelism solely for the \emph{Diffuse} stage, since \emph{Encode} requires no parallelism and the \emph{Decode} stage has very small model parameters, rendering model parallelism highly inefficient. When the \emph{Diffusion model} cannot fit on a single GPU, we first determine the minimal model-parallel degree $k_{min}$ such that, under the maximum load, the per-GPU shard of the \emph{Diffusion model} fits within one GPU’s memory. We then carry out the \emph{placement plan} allocation and \emph{dispatch plan} solving at the granularity of $k_{min}$ GPUs (rather than a single GPU), which leaves all our methods unchanged.

\end{document}